\documentclass[IEEE journal]{IEEEtran}
\usepackage{algorithmic}
\usepackage{url}
\usepackage{optidef} 
\usepackage{cite}
\usepackage{amsmath,amssymb,amsfonts}
\usepackage{bm}
\usepackage{algorithm}
\usepackage{algorithmic}

\usepackage{float}
\usepackage{amsthm}
\usepackage{graphics}
\usepackage{graphicx}
\usepackage{textcomp}
\usepackage{xcolor}
\usepackage{dsfont}
\usepackage{comment}
\usepackage{rotating}
\usepackage{makecell}
\usepackage{multirow}

\theoremstyle{plain}

\newtheorem{proposition}{Proposition}

\def\BibTeX{{\rm B\kern-.05em{\sc i\kern-.025em b}\kern-.08em
    T\kern-.1667em\lower.7ex\hbox{E}\kern-.125emX}}
\usepackage{subcaption}    
	\definecolor{mygreen}{rgb}{0.01, 0.75, 0.24}    
\providecolor{added}{rgb}{0,0,1}
\providecolor{deleted}{rgb}{1,0,0}
\newcommand{\added}[1]{{\color{added}{}#1}}

\usepackage{bbm}

\begin{document}

\title{Green Cell-Free Massive MIMO for ISAC: Joint Cloud,
Fronthaul and Radio Resource Allocation
}

\author{%
Zinat~Behdad,
\"Ozlem~Tu\u{g}fe~Demir,
Ki~Won~Sung,
Pei~Xiao,
and~Cicek~Cavdar%
\thanks{Z. Behdad, K. W. Sung, and C. Cavdar are with the Department of Communication Systems, KTH Royal Institute of Technology, Stockholm, Sweden (e-mail: \{zinatb, sungkw, cavdar\}@kth.se). \"O. T. Demir is with the Department of Electrical and Electronics Engineering, Bilkent University, Ankara, T\"urkiye (e-mail: ozlemtugfedemir@bilkent.edu.tr). P. Xiao is with 5GIC \& 6GIC, the Institute for Communication Systems (ICS), University of Surrey, Guildford, United Kingdom (e-mail: p.xiao@surrey.ac.uk).}%
}

\maketitle
\begin{abstract}
In this paper, we develop a cross-layer end-to-end (E2E) resource orchestration framework for green cell-free massive multiple-input multiple-output (CF-mMIMO) integrated sensing and communication (ISAC) systems with distributed multi-target detection. We propose a distributed sensing approach in which receive access points (RX-APs) compute local test statistics, which are aggregated at the cloud using weights based on sensing interference and channel quality. We derive the local maximum a posteriori ratio test (MAPRT) detectors under fully informed and partially informed operation, representing different levels of transmit-signal information at the RX-APs. We further characterize their processing and fronthaul requirements and derive a network power model incorporating radio transmission, AP and cloud processing, and fronthaul infrastructure. We formulate a
mixed-integer non-convex problem that jointly optimizes transmit powers, AP modes, user equipment (UE) and sensing-area associations, RX-AP assignments, and active fronthaul and cloud resources subject to communication, sensing, power, processing-capacity, and fronthaul constraints. A two-stage iterative algorithm based on Big-$M$ reformulation, convex--concave programming, penalty-based relaxation, and structured discrete recovery is developed. Numerical results show that the proposed E2E framework reduces total power by up to $50\%$ compared with transmit-power-only optimization and by approximately $11-17\%$ compared with joint radio optimization under full coordination, while maintaining detection probabilities above $0.95$. 
\end{abstract}
\vspace{-4mm}
\section{Introduction}
\vspace{-2mm}
Integrated sensing and communication (ISAC) and cell-free massive multiple-input multiple-output (CF-mMIMO) are considered as key technologies for sixth-generation (6G) wireless networks \cite{liu2022integrated}. By deploying a large number of geographically distributed access points (APs), CF-mMIMO systems can support multi-static sensing without the full-duplex operation required by mono-static sensing systems. Furthermore, their integration with cloud radio access network (C-RAN) architectures enables centralized coordination and signal processing across APs, making CF-mMIMO a promising platform for implementing ISAC functionalities on general-purpose processors (GPPs) \cite{demir2023cell,behdad2024}.

Despite these advantages, scalability remains a major challenge. 
Several studies adopted the principles of user-centric communication [4], [5] and target-centric sensing [6] to develop scalable CF-mMIMO ISAC architectures. In \cite{chu2023integrated}, an OFDM-based user-centric framework is developed 
and the communication--sensing trade-off under different power-splitting strategies is investigated. The authors in \cite{femenias2025scalable} propose a scalable framework based on user- and target-centric AP cooperation, scalable precoding, fractional power allocation, and maximum a posteriori ratio test (MAPRT)-based target detection. This framework is extended in \cite{femenias2025isac} to account for transceiver hardware impairments through impairment-aware scalable precoders and MAPRT-based detectors. 

Multi-target detection introduces further complexity because the receiver must distinguish reflections originating from multiple targets. Centralized joint detection based on multiple-hypothesis testing can provide high detection accuracy \cite{zhang2023multi-target}, but its computational and fronthaul requirements may become prohibitive in large-scale systems. Distributed sensing provides a more scalable alternative: the receive APs locally process their observations and forward compressed information or local test statistics to the cloud. Although this approach reduces centralized processing requirements and facilitates real-time detection, it raises new challenges concerning the aggregation of distributed measurements and the preservation of detection performance.

Distributed processing in CF-mMIMO ISAC has been investigated from several perspectives. In \cite{chowdhury2026joint}, centralized and distributed generalized likelihood ratio test (GLRT) detectors are developed for dynamic time division duplex (TDD) CF-mMIMO ISAC, and their detection-performance--complexity trade-off is characterized under bidirectional communication. In \cite{hong2026team}, distributed ISAC precoding is performed across multiple central processing units (CPUs) using local channel state information and partially shared side information over non-ideal inter-CPU fronthaul links. Moreover, \cite{elfiatoure2025multiple} jointly considers AP operation and power control to improve communication spectral efficiency subject to sensing constraints. These studies address important aspects of distributed processing, but do not investigate cloud-based aggregation of local sensing measurements for scalable multi-target detection.

The integration of sensing into communication networks increases power consumption not only through additional radio transmission but also through the associated fronthaul and processing loads. End-to-end (E2E) energy-efficient design has been widely studied in communication networks \cite{alabbasi2018optimal,masoudi2022energy,demir2022cell,demir2023cell,yang2019efficient}, but remains relatively unexplored in ISAC systems. In particular, \cite{demir2023cell} formulates an E2E power-minimization problem for a fully virtualized CF-mMIMO system under the O-RAN architecture, jointly considering radio, fronthaul, and processing resources, but does not include sensing functionalities. More recently, \cite{behdad2026energy} investigates joint radio and processing energy minimization for CF-mMIMO ISAC with ultra-reliable low-latency communication (URLLC) requirements. However, it assumes centralized sensing, excludes fronthaul power and capacity requirements, and does not address scalability in multi-target scenarios.

Building on our previous work on distributed sensing in single-target scenarios \cite{Zou2024Distributed}, this paper addresses multi-target detection, which entails higher resource demands and more stringent performance requirements. In the conference version \cite{behdad2025detecting}, we investigated the aggregation of local test statistics from receive APs (RX-APs) at the cloud and developed a power-allocation algorithm that jointly maximizes the minimum communication and sensing signal-to-interference-plus-noise ratio (SINR). In this extended work, we jointly orchestrate radio, fronthaul, and cloud-processing resources to minimize the total network power consumption. Unlike \cite{femenias2025isac}, where computational complexity and fronthaul signaling are addressed through scalability analysis and predefined AP clusters, the proposed framework explicitly incorporates processing capacity and fronthaul constraints into the resource-orchestration problem.

\vspace{-4.5mm}
\subsection{Contributions}
\vspace{-1mm}
We develop a cross-layer framework for energy-efficient distributed multi-target sensing in CF-mMIMO ISAC systems. The main contributions are summarized as follows: 
\begin{itemize}
    \item We propose a distributed multi-target sensing approach in which RX-APs compute local MAPRT statistics that are aggregated at the cloud using weights determined by the sensing interference and channel quality. The local detectors are derived for the fully informed systems (FIS) and partially informed systems (PIS), capturing different levels of transmit-signal information available at the RX-APs.
    \item We characterize the radio, fronthaul, and processing requirements of the proposed architecture. In particular, we derive the AP- and cloud-processing loads, quantify the FIS and PIS detector complexities, formulate their corresponding fronthaul requirements under functional split Option~7.2, and derive the E2E network power-consumption model.

    \item We formulate a joint E2E resource orchestration problem that minimizes the total radio, fronthaul, and cloud power consumption subject to communication and sensing SINR requirements, radio power, processing, and fronthaul capacity constraints. 
    We develop a two-stage iterative solution for the resulting mixed-integer non-convex problem. The communication constraints are reformulated in second-order cone (SOC) form, while the sensing and RX-processing constraints are handled using Big-$M$, feasible point pursuit successive convex approximation (FPP-SCA), and convex-concave programming (CCP) approximations. 
\end{itemize}



\vspace{-3mm}
\section{System Model}\label{sec:system model}
\vspace{-1mm}
We consider a CF-mMIMO ISAC system with $L$ distributed APs connected to a central cloud via fronthaul links. 
Each AP is equipped with $M$ antennas arranged in a horizontal uniform linear array (ULA) with half-wavelength spacing. The network serves $K$ single-antenna UEs and $S$ non-overlapping sensing service areas (SSAs) \footnote{We assume synchronized APs and a cyclic-prefix duration sufficient to cover the selected communication and bistatic sensing path delays.}.  
The system operates under a virtualized C-RAN architecture, where baseband processing functions are implemented at the cloud using GPPs. 
We adopt functional split Option 7.2, in which lower-physical-layer (PHY) functions, including precoding and combining, are performed at the APs, while higher-layer baseband processing is centralized in the cloud \cite{3gpp38816}. This split enables efficient coordination across APs while reducing fronthaul overhead compared to fully centralized processing, and is also suitable for distributed sensing where local signal processing is required at RX-APs.

Each AP operates in one of three modes: i) ISAC transmit mode (TX-AP), where it simultaneously supports downlink communication and sensing transmission; ii) sensing receive mode (RX-AP), where it captures target reflections and performs local processing; or iii) idle mode. Accordingly, we define the sets $\mathcal{L}_{\rm tx}$, $\mathcal{L}_{\rm rx}$, and $\mathcal{L}_{\rm idle}$ as the collections of TX-APs, RX-APs, and idle APs, respectively. 

To enable scalable operation, we adopt a user-centric and target-centric association strategy. Specifically, each UE $k \in \mathcal{K}=\{1,\ldots,K\}$ is served by a subset of TX-APs denoted by $\mathcal{M}_{k} \subseteq \mathcal{L}_{\rm tx}$. Similarly, each SSA $s \in \mathcal{S}=\{1,\ldots,S\}$ is supported by subsets of TX-APs and RX-APs, denoted by $\mathcal{T}_{s} \subseteq \mathcal{L}_{\rm tx}$ and $\mathcal{R}_{s} \subseteq \mathcal{L}_{\rm rx}$, respectively. The number of RX-APs assigned to each SSA is $R$. From the AP perspective, $\mathcal{U}_{l} \subseteq \mathcal{K}$ and $\mathcal{S}_{l} \subseteq \mathcal{S}$ denote the sets of UEs and SSAs served by TX-AP $l$.

\vspace{-3mm}
\subsection{Downlink Communication and Sensing Transmission}
\vspace{-1mm}
TX-APs simultaneously serve communication UEs and SSAs. The transmit signal at AP $l$ is given by 
\begin{align}
    \textbf{x}_l[m]&= \sum_{k\in \mathcal{U}_l}  \sqrt{p_{k,l}}\,\textbf{w}_{k,l}\, s_k[m]\,+\sum_{s\in \mathcal{S}_l}\sqrt{q_{s,l}}\,  \boldsymbol{\omega}_{s,l}\, r_s[m],
\end{align}
where $p_{k,l}$ and $q_{s,l}$ are the power coefficients for UE $k$ and SSA $s$, respectively. The unit-power communication and sensing symbols at time-frequency index $m$ are denoted by $s_k[m]$ and $r_s[m]$. The precoding vectors $\textbf{w}_{k,l}$ and $\boldsymbol{\omega}_{s,l}$ are designed based on channel state information (CSI) using local partial minimum mean-squared error (LP-MMSE) precoding for communication and maximum ratio transmission (MRT) for sensing \cite{cell-free-book}.

Let $\textbf{h}_{kl}\in \mathbb{C}^M$ denote the channel vector from AP $l$ to UE $k$. It is modeled as spatially correlated Rician fading, as $\textbf{h}_{kl} =e^{j\psi_{kl}}\bar{\textbf{h}}_{kl}\,+\,\tilde{\textbf{h}}_{kl}$, where $\bar{\textbf{h}}_{kl}$ represents the deterministic line-of-sight (LOS) component with random phase $\psi_{kl}\sim \mathcal{U}[0,2\pi)$, and $\tilde{\textbf{h}}_{kl}\sim \mathcal{CN}(\textbf{0},\tilde{\textbf{R}}_{kl})$ denotes the non-line-of-sight (NLOS) component with spatial correlation matrix $\tilde{\textbf{R}}_{kl}$. Both components incorporate large-scale fading effects.

The normalized LP-MMSE precoding vector for UE $k$ at TX-AP $l$ is given by
$\textbf{w}_{k,l}= \frac{\overline{\textbf{w}}_{k,l}}{\sqrt{\mathbb{E}\{\Vert \overline{\textbf{w}}_{k,l}\Vert^2\}}}$,
where $\overline{\textbf{w}}_{k,l}=p_k^{\text{ul}} \left( \sum_{i\in {\mathcal{U}}_l} p_i^{\text{ul}}\left(\hat{\textbf{h}}_{il}\hat{\textbf{h}}_{il}^H + \textbf{Z}_{il}\right)+ \sigma_n^2 \textbf{I}_M\right)^{-1} \hat{\textbf{h}}_{kl}$, 
and $p_i^{\text{ul}}$ is the uplink pilot power, $\hat{\textbf{h}}_{kl}$ is the linear minimum mean-squared error (LMMSE) channel estimate \cite{wang2020uplink}, $\textbf{Z}_{il}$ is the channel estimation error covariance, and $\sigma_n^2$ is the noise variance.


For sensing, MRT precoding is adopted as
$\boldsymbol{\omega}_{s,l}=\frac{\boldsymbol{\mathfrak{h}}^*_{sl}}{\left\lVert \boldsymbol{\mathfrak{h}}_{sl}\right\rVert}$,
where $\boldsymbol{\mathfrak{h}}_{sl}=\sqrt{\beta_{sl}}\textbf{a}(\varphi_{s,l},\vartheta_{s,l})$ represents the LOS channel between AP $l$ and SSA $s$, with $\beta_{sl}$ denoting the channel gain. $\textbf{a}(\varphi_{s,l},\vartheta_{s,l})$  
is the array response vector with the azimuth and elevation angles from TX-AP $l$ to the target location $s$, denoted by $\varphi_{s,l}$ and $\vartheta_{s,l}$, respectively. NLOS components are neglected due to the severe attenuation in the two-way sensing channel \cite{buzzi2024scalability}.
The received signal at UE $k$ is expressed as
\begin{align}\label{eq:yk}
    y_{k}[m] 
    &= \sum_{l\in\mathcal{M}_{k}} \sqrt{p_{k,l}} \textbf{h}_{kl}^H \textbf{w}_{k,l}\, s_k[m]\nonumber\\
    &\quad+ \sum_{l\in\mathcal{L}_{\rm tx}}\sum_{j\in \mathcal{U}_l\setminus \{k\}} \sqrt{p_{j,l}} \textbf{h}_{kl}^H \textbf{w}_{j,l}\, s_j[m] \nonumber\\
    &\quad+ \sum_{l\in\mathcal{L}_{\rm tx}} \sum_{s\in \mathcal{S}_l} \sqrt{q_{s,l}} \textbf{h}_{kl}^H \boldsymbol{\omega}_{s,l}\, r_s[m] 
    + n_k[m],
\end{align}
where $n_k[m]\sim \mathcal{CN}(0,\sigma_n^2)$ is additive independent Gaussian noise. The first term in \eqref{eq:yk} is the desired signal, while the second and third terms represent the interference due to other communication UEs and sensing signals, respectively.

Let $\tau_c$ and $\tau_d=\tau_c-\tau_p$ denote the number of total symbols and data symbols per coherence block, where $\tau_p$ is the number of pilot symbols in each coherence block. The achievable spectral efficiency (SE) of UE $k$ is
\begin{align}
    \mathsf{SE}_k= \frac{\tau_d}{\tau_c}\log_2\left(1+\mathsf{SINR}^{\text{comm}}_k\right),
\end{align}
where $\mathsf{SINR}^{\text{comm}}_k$ is the effective SINR based on use-and-then-forget bound, given by \cite{demir2023cell}
\begin{align}
    \mathsf{SINR}^{\text{comm}}_k = \frac{\vert \textbf{a}_k^T \boldsymbol{\rho}_k\vert ^2}{\sum_{j\in\mathcal{K}}\boldsymbol{\rho}_j^T \textbf{B}_{kj}\boldsymbol{\rho}_j+ \sum_{s\in \mathcal{S}}\boldsymbol{q}_s^T \textbf{C}_{ks}\boldsymbol{q}_s + \sigma_n^2},
\end{align}
where $\boldsymbol{\rho}_k=[\sqrt{p_{k,1}}, \ldots,\sqrt{p_{k,L_{\rm tx}}}]^{T}$ and $\boldsymbol{q}_s=[\sqrt{q_{s,1}}, \ldots,\sqrt{q_{s,L_{\rm tx}}}]^{T}$ collect the power coefficients, where $L_{\rm tx}=|\mathcal{L}_{\rm tx}|$. The  $\textbf{a}_k$, $\textbf{B}_{kj}$, and $\textbf{C}_{ks}$ characterize the desired signal, interference, and sensing-induced interference, respectively, with entries defined as
\begin{align}
    &[\textbf{a}_k]_l = \mathbb{E}\left\{\textbf{h}_{kl}^H \textbf{w}_{k,l}\right\},\quad [\textbf{C}_{ks}]_{ll'} = \mathbb{E}\left\{\textbf{h}_{kl}^H \boldsymbol{\omega}_{s,l} \boldsymbol{\omega}_{s,l'}^H \textbf{h}_{kl'} \right\} \\
    &[\textbf{B}_{kj}]_{ll'} =
    \begin{cases}
     \mathbb{E}\left\{\textbf{h}_{kl}^H \textbf{w}_{j,l} \textbf{w}_{j,l'}^H \textbf{h}_{kl'}\right\}, & k\neq j \\
     \mathbb{E}\left\{\textbf{h}_{kl}^H \textbf{w}_{j,l} \textbf{w}_{j,l'}^H \textbf{h}_{kl'}\right\} - [\textbf{a}_k]_l[\textbf{a}_k]_{l'}^*, & k=j
    \end{cases} 
\end{align}
following from the assumption that the communication and sensing symbols are zero-mean, unit-power, and statistically independent across UEs and SSAs.

\subsection{Sensing Reception}
RX-APs receive reflected signals from targets within their assigned SSAs and compute local test statistics, which are forwarded to the cloud for final decision-making. Although an RX-AP may receive reflections from multiple targets, it computes local test statistics only for its assigned SSAs, thereby reducing computational complexity and limiting processing to the relevant sensing tasks.\footnote{We assume that the TX--RX channel in the absence of targets is known and communication interference is negligible; its impact is left for future work.}

To enhance the received signal quality, we employ normalized maximum ratio combining (MRC). Consider an RX-AP $r$ assigned to SSA $s$, i.e., $r\in\mathcal{R}_s$. The corresponding combining vector is defined as $\textbf{v}_{s,r}=\frac{\boldsymbol{\mathsf{h}}_{sr}}{\left\lVert \boldsymbol{\mathsf{h}}_{sr}\right\Vert}=\frac{\textbf{a}(\phi_{s,r},\theta_{s,r})}{\Vert \textbf{a}(\phi_{s,r},\theta_{s,r})\Vert}$, where
$\boldsymbol{\mathsf{h}}_{sr}=\sqrt{\overline{\beta}_{sr}}\textbf{a}(\phi_{s,r},\theta_{s,r})\in\mathbb{C}^{M}$ denotes the LOS channel between SSA $s$ and RX-AP $r$, with channel gain $\overline{\beta}_{sr}$ and azimuth and elevation angles $\phi_{s,r}$ and $\theta_{s,r}$, respectively.
The received signal at RX-AP $r$ for SSA $s$ is given by
\begin{align} \label{eq:y_ls}
    y_{s,r} [m]\!&= \!\textbf{v}_{s,r}^H\!\sum_{l\in \mathcal{L}_{\rm tx}}\!\alpha_{s,r,l}\textbf{G}_{s,r,l}\textbf{x}_l[m] \nonumber\\
    &+\!\textbf{v}_{s,r}^H\!\sum_{t\in \mathcal{S}\setminus \{s\} }\!\sum_{l\in \mathcal{L}_{\rm tx}}\!\alpha_{t,r,l}\textbf{G}_{t,r,l}\textbf{x}_l[m]
   + \!\textbf{v}_{s,r}^H\textbf{n}_r[m]
\end{align}
where $\alpha_{s,r,l}\sim\mathcal{CN}(0,1)$ is the normalized bistatic radar cross section (RCS) coefficient following the Swerling-I model. $\textbf{G}_{s,r,l}=\sqrt{\beta_{s,r,l} }\textbf{a}(\phi_{s,r},\theta_{s,r})\textbf{a}^{T}(\varphi_{s,l},\vartheta_{s,l})\in \mathbb{C}^{M \times M}$ is the two-way channel matrix from TX-AP $l$ to RX-AP $r$ via the SSA $s$, where $\beta_{s,r,l}$ denotes the corresponding channel gain, including both path loss and the target RCS variance, following the radar range equation for bi-static sensing \cite[Chap.~2]{richards2010principles}. 

For notational simplicity, we define $\textbf{G}_{s,r}[m]=\begin{bmatrix}\textbf{G}_{s,r,1}\textbf{x}_1[m]&\cdots&\textbf{G}_{s,r,L_{\rm tx}}\textbf{x}_{L_{\rm tx}}[m]\end{bmatrix}$, and let $\boldsymbol{\alpha}_{s,r}$ collect the coefficients $\alpha_{s,r,l}$. Then,
\begin{align}\label{eq:y_s,r}
     y_{s,r}\! [m]
     &\!= \!\textbf{v}_{s,r}^H\!\textbf{G}_{s,r}[m]\boldsymbol{\alpha}_{s,r}
     \!\!+\!\! \textbf{v}_{s,r}^H\!\!\!\sum_{t\in\mathcal{S}\!\setminus \{s\}}\!\!\!\!\textbf{G}_{t,r}[m]\boldsymbol{\alpha}_{t,r}
    \! + \!\textbf{v}_{s,r}^H\textbf{n}_r[m]\nonumber\\
    & = g_{s,r}[m]+ I_{s,r}[m]+n_{s,r}[m],
\end{align}
where the first term is the desired signal, while the second and third terms are the target-related interference and noise, respectively.
The signals $y_{s,r}[m]$ form the basis for local test statistics computed at each RX-AP, as detailed in Section~\ref{sec:detector}. The quality of these signals directly impacts the reliability of target detection. Accordingly, we characterize the sensing performance by SINR. Forming $\overline{\boldsymbol{\rho}}_l\!=\!\begin{bmatrix} \!\sqrt{p_{1,l}}\!&\! \!\!\ldots\!\!\! \!& \!\sqrt{p_{K,l}}\end{bmatrix}^T \!\!\!\!\in \!\!\mathbb{R}^{K}$ and $\overline{\boldsymbol{q}}_l\!=\!\!\begin{bmatrix} \sqrt{q_{1,l}}&\!\! \ldots\! & \!\!\!\sqrt{q_{S,l}}\end{bmatrix}^T\!\!\!\in\!\mathbb{R}^{S}$ as the the communication and sensing power coefficient vectors at TX-AP $l$, the sensing SINR at RX-AP $r$ for SSA $s$ over $\tau_s$ channel uses (under the assumption of independent RCS coefficients) is given by
\eqref{eq:sensing_snr}, at the top of next page,
\begin{figure*}
\begin{align}
\label{eq:sensing_snr}
    \mathsf{SINR}^{\text{sens}}_{s,r} = \frac{ \sum_{m=1}^{\tau_s}\sum_{l\in \mathcal{L}_{\rm tx}}\vert\textbf{d}_{s,r,l,m}^T\overline{\boldsymbol{\rho}}_l+\textbf{e}_{s,r,l,m}^T\overline{\boldsymbol{q}}_l\vert^2}{\sum_{m=1}^{\tau_s}\sum_{t\in \mathcal{S}\setminus \{s\}}\sum_{l\in \mathcal{L}_{\rm tx}}\vert\textbf{f}_{s,t,r,l,m}^T\overline{\boldsymbol{\rho}}_l+\textbf{g}_{s,t,r,l,m}^T\overline{\boldsymbol{q}}_l\vert^2+
    \tau_s\sigma_n^2}.
\end{align} 
   \hrulefill
   \vspace{-6mm}
\end{figure*}
where the entries of $\textbf{d}_{s,r,l,m}$, $\textbf{e}_{s,r,l,m}$, $\textbf{f}_{s,t,r,l,m}$, and $\textbf{g}_{s,t,r,l,m}$ are defined as
\begin{align}
    &[\textbf{d}_{s,r,l,m}]_{k} = \textbf{v}_{s,r}^H\,\textbf{G}_{s,r,l}\textbf{w}_{k,l}s_k[m], \quad\forall k\in \mathcal{K}\\
    &[\textbf{e}_{s,r,l,m}]_{t} = \textbf{v}_{s,r}^H\,\textbf{G}_{s,r,l}\boldsymbol{\omega}_{t,l}r_{t}[m], \quad \forall t\in \mathcal{S}\\
    &[\textbf{f}_{s,t,,r,l,m}]_{k}=\textbf{v}_{s,r}^H\,\textbf{G}_{t,r,l}\textbf{w}_{k,l}s_k[m], \quad\forall k\in \mathcal{K}\\ 
  &[\textbf{g}_{s,t,r,l,m}]_{t'}=\textbf{v}_{s,r}^H\,\textbf{G}_{t,r,l}\boldsymbol{\omega}_{t',l}r_{t'}[m], \quad \forall t'\in \mathcal{S}
.\end{align}
We form the concatenated power vector $\boldsymbol{\rho}\!=\!\begin{bmatrix}\overline{\boldsymbol{\rho}}_1^T &\overline{\boldsymbol{q}}_1^T &\cdots &\overline{\boldsymbol{\rho}}_{L_{\rm tx}}^T &\overline{\boldsymbol{q}}_{L_{\rm tx}}^T\end{bmatrix}^T\!\!\in \!\mathbb{R}^{(K+S)L_{\rm tx}}$ and the block diagonal matrices $\boldsymbol{\mathsf{A}}_{s,r}$ and $\boldsymbol{\mathsf{B}}_{s,r}$ whose $l$-th blocks are
\begin{align}
    [\boldsymbol{\mathsf{A}}_{s,r}]_{l}
    &\!=\!
    \sum_{m=1}^{\tau_s}
    \Re\left\{\!
    \begin{bmatrix}
        \textbf{d}_{s,r,l,m}\\
        \textbf{e}_{s,r,l,m}
    \end{bmatrix}
    \begin{bmatrix}
        \textbf{d}_{s,r,l,m}^{H} & \textbf{e}_{s,r,l,m}^{H}
    \end{bmatrix}
    \!\right\},\\
    [\boldsymbol{\mathsf{B}}_{s,r}]_{l}
    &\!=\!
    \sum_{m=1}^{\tau_s}\!
    \sum_{t\in \mathcal{S}\setminus \{s\}}
    \!\!\!\!\Re\!\left\{\!\!
    \begin{bmatrix}
        \textbf{f}_{s,t,r,l,m}\\
        \textbf{g}_{s,t,r,l,m}
    \end{bmatrix}\!\!\!
    \begin{bmatrix}
        \textbf{f}_{s,t,r,l,m}^{H} & \!\!\textbf{g}_{s,t,r,l,m}^{H}
    \end{bmatrix}
    \!\right\}
\end{align}
The sensing SINR in \eqref{eq:sensing_snr} is  then written in compact quadratic form as
\begin{align}
    \mathsf{SINR}^{\text{sens}}_{s,r}
    =
    \frac{\boldsymbol{\rho}^T \boldsymbol{\mathsf{A}}_{s,r}\boldsymbol{\rho}}
    {\boldsymbol{\rho}^T \boldsymbol{\mathsf{B}}_{s,r}\boldsymbol{\rho}+\tau_s\sigma_n^2}.
\end{align}

\vspace{-2mm}
\section{Distributed Target Detection}
\vspace{-1mm}
\label{sec:detector}
To detect multiple targets, we formulate a set of binary hypothesis tests for each SSA. The hypothesis testing problem for SSA $s$ at RX-AP $r \in \mathcal{R}_s$, over $m = 1,\ldots,\tau_s$, is given by
\begin{align}\label{hypothesis}
   &\mathcal{H}_{s,0}:~ y_{s,r}[m] = I_{s,r}[m] + n_{s,r}[m], \nonumber\\
   &\mathcal{H}_{s,1}:~ y_{s,r}[m] = g_{s,r}[m] + I_{s,r}[m] + n_{s,r}[m],
\end{align}
where $\mathcal{H}_{s,0}$ denotes target absence in SSA $s$, while $\mathcal{H}_{s,1}$ denotes target presence.

Each RX-AP computes local test statistics only for its assigned SSAs. Let $T_{s,r}$ denote the local test statistic for SSA $s$ at RX-AP $r$. These local test statistics are aggregated at the cloud to form the final test statistic as
\begin{align}
    T_s = \sum_{r\in \mathcal{R}_{s}} w_{s,r} T_{s,r}.
\end{align}
The final test statistic is then compared with the decision threshold  $\lambda_{d,s}$, which is selected empirically to meet a target false alarm probability. 
The weights $w_{s,r}$ are designed to reflect the sensing quality at each RX-AP, defined as
\begin{align}\label{eq:weights}
     w_{s,r} = \frac{\overline{w}_{s,r}^v }{\sum_{r'\in \mathcal{R}_s}\overline{w}_{s,r'}^v},
\end{align}
which are based on the signal-to-interference ratio (SIR) at RX-AP $r$ for SSA $s$, denoted by $\overline{w}_{s,r}$ and given by 
\begin{align}
\overline{w}_{s,r} = \frac{\overline{\beta}_{s,r}\left| \textbf{v}_{s,r}^H \textbf{a}(\phi_{s,r},\theta_{s,r})\right|^2 }{\sum_{t\in \mathcal{S}\setminus \{s\}}\overline{\beta}_{t,r}\left| \textbf{v}_{s,r}^H \textbf{a}(\phi_{t,r},\theta_{t,r})\right|^2},
\end{align}
and weighting exponent $v$ that controls the contribution of each RX-AP. The fraction in \eqref{eq:weights} determines the channel quality at RX-AP $r$ compared to the other serving RX-APs.

We derive the local test statistics for a distributed MAPRT detector under two levels of side information at the RX-APs: i) a fully informed system (FIS) and ii) a partially informed system (PIS). In FIS, each RX-AP knows the transmitted sensing signals, improving detection performance at the cost of higher cloud-to-AP fronthaul signaling. In PIS, the RX-APs only have statistical information about the transmitted signals, reducing signaling requirements but potentially degrading detection performance \cite{Zou2024Distributed}. Thus, FIS and PIS represent a trade-off between sensing performance and fronthaul overhead.

Before deriving the local test statistics, we define $c_{s,r,l}[m]
\triangleq
\sqrt{\beta_{s,r,l}}\,
\mathbf{a}^{T}(\varphi_{s,l},\vartheta_{s,l})
\mathbf{x}_{l}[m],$
and let $\mathbf{c}_{s,r}[m]\in\mathbb{C}^{L_{\rm tx}}$ collect the conjugates of these terms across the TX-APs. Accordingly, \eqref{eq:y_s,r} can be rewritten as
\begin{align}
y_{s,r}[m]
&=\!\!
\sqrt{M}\mathbf{c}_{s,r}^{H}[m]\boldsymbol{\alpha}_{s,r}\!
+\!\mathbf{v}_{s,r}^{H}\!\!\!\sum_{t\in\mathcal{S}\setminus\{s\}}\!\!\!
\mathbf{a}(\phi_{t,r},\theta_{t,r})
\mathbf{c}_{t,r}^{H}[m]\boldsymbol{\alpha}_{t,r}
\nonumber\\[-2mm]
&+n_{s,r}[m],
\end{align}
where $\mathbf{v}_{s,r}^{H}\mathbf{a}(\phi_{s,r},\theta_{s,r})=\sqrt{M}$ under MRC. The second term represents interference from other targets. Incorporating this term into the detector would require joint multi-target hypothesis modeling and would significantly increase detector complexity. We therefore neglect it in the derivation, while retaining it in the numerical received-signal model and calibrating the threshold from the corresponding $\mathcal{H}_{s,0}$ samples. Thus, the false-alarm probability is controlled under the full model, and the remaining mismatch mainly affects detection probability. 

In the FIS scenario, RX-APs have full access to $\{\textbf{c}_{s,r}[m]\}$ vectors, array response vector, and the RCS correlation matrix corresponding to all SSA $s$ denoted by $\textbf{R}_{s,r}^\textrm{rcs}\in \mathbb{C}^{L_\textrm{tx}\times L_\textrm{tx}}$. The local test statistic is obtained as \footnote{The local test statistics are suboptimal due to unaccounted sensing interference, which will be addressed in future work.}
\begin{align}
    T^{\rm FIS}_{s,r}= \ln\left(\frac{\max_{\boldsymbol{\alpha}_{s,r}} \prod_{m=1}^{\tau_s}p\left(y_{s,r} [m] |\boldsymbol{\alpha}_{s,r},\mathcal{H}_{s,1}\right) p\left(\boldsymbol{\alpha}_{s,r}\right)}{ \prod_{m=1}^{\tau_s}p\left(y_{s,r} [m]|\mathcal{H}_{s,0}\right)}\right)\nonumber
\end{align}
where $p\left(y_{s,r} [m] |\boldsymbol{\alpha}_{s,r},\mathcal{H}_{s,1}\right)$ is the probability distribution function (PDF) of the received signal for a given $\boldsymbol{\alpha}_{s,r}$ under hypothesis $\mathcal{H}_{s,1}$ and $p\left(\boldsymbol{\alpha}_{s,r}\right)$ is the PDF of RCS variables.
This leads to the closed-form expression
\begin{align}\label{eq:T_r_FIS}
    T^{\rm FIS}_{s,r} = \textbf{a}_{s,r}^H\,\textbf{C}_{s,r}^{-1}\,\textbf{a}_{s,r},
\end{align}
where $\textbf{a}_{s,r}\in \mathbb{C}^{L_{\rm tx}}$ and $\normalfont\textbf{C}_{s,r}\in \mathbb{C}^{L_{\rm tx}\times L_{\rm tx}}$ are

\begin{align}
    &\textbf{a}_{s,r}= \sqrt{M}\sum_{m=1}^{\tau_s} \textbf{c}_{s,r}[m] \,y_{s,r}[m],\label{eq:vecA}\\
    &\normalfont\textbf{C}_{s,r}\!= M \!\sum_{m=1}^{\tau_s} \!\textbf{c}_{s,r}[m] \textbf{c}_{s,r}^H[m]\!+ \!\sigma_n^2 (\textbf{R}_{s,r}^\textrm{rcs})^{-1}.\label{eq:matrixCinv}
\end{align}


In the PIS scenario, RX-APs have access only to the statistics of the unknown RCS variables and statistical information about the transmit signals. 
The local test statistic is then obtained by jointly estimating the RCS vector $\boldsymbol{\alpha}_{s,r}$ and $\{\textbf{c}_{s,r}[m]\}$ as expressed in \eqref{eq:T_PIS}, on the top of the next page.
The solution is obtained via an alternating iterative procedure, where $\{\textbf{c}_{s,r}[m]\}$ and $\boldsymbol{\alpha}_{s,r}$ are updated using \eqref{c_r} and \eqref{eq:alpha_r}, respectively where $\mathbf{R}_{s,r}\in \mathbb{C}^{L_{\rm tx} \times L_{\rm tx}}$ is the covariance matrix of vector $\textbf{c}_{s,r}[m]$ \cite{Zou2024Distributed}, which can be expressed as  
\begin{align}
\mathbf{R}_{s,r}
&=
\mathbf{D}_{s,r}\widetilde{\mathbf{R}}_{s}\mathbf{D}_{s,r},\label{eq:R_s,r}
\end{align}
with $\mathbf{D}_{s,r}
=
\operatorname{diag}\!\left(
\sqrt{\beta_{s,r,1}},\ldots,\sqrt{\beta_{s,r,L_{\rm tx}}}
\right).
$ and $\widetilde{\mathbf{R}}_{s}
=
\sum_{k=1}^{K}\mathbf{u}_{s,k}\mathbf{u}_{s,k}^{H}
+
\sum_{t=1}^{S}\mathbf{v}_{s,t}\mathbf{v}_{s,t}^{H}$ 
where $\mathbf{u}_{s,k}$ and $\mathbf{v}_{s,t}$ are the concatenated vectors of
$u_{s,k,l}=\sqrt{p_{k,l}}\,
\mathbf{a}^{T}(\varphi_{s,l},\vartheta_{s,l})\mathbf{w}_{k,l}$ and 
$v_{s,t,l}=\sqrt{q_{t,l}}\,
\mathbf{a}^{T}(\varphi_{s,l},\vartheta_{s,l})\boldsymbol{\omega}_{t,l}$ across the active TX-APs, respectively.

The corresponding local test statistic at iteration $i$ is given by \eqref{eq:fpartlyinformed}. Due to space limitations, we refer the reader to \cite[Algorithm 1]{Zou2024Distributed} for detailed derivations and implementation.
\begin{figure*}
\begin{align}
   & T^{\rm PIS}_{s,r}= \ln\left(\frac{\max_{\boldsymbol{\alpha}_{s,r},\{\textbf{c}_{s,r}[m]\}_{\!m=1}^{\!\tau_s}} \prod_{m=1}^{\tau_s}p\left(y_{s,r} [m] |\boldsymbol{\alpha}_{s,r},\textbf{c}_{s,r}[m],\mathcal{H}_{s,1}\right) p\left(\boldsymbol{\alpha}_{s,r}\right)p\left(\left\{\textbf{c}_{s,r}[m]\right\}_{m=1}^{\tau_s}|\mathcal{H}_{s,1}\!\right)}{ \prod_{m=1}^{\tau_s}p\left(y_{s,r} [m]|\mathcal{H}_{s,0}\right)}\right) \label{eq:T_PIS}\\
&\hat{\textbf{c}}_{s,r}^{(i)}[m]=\left(\left (M\hat{\boldsymbol{\alpha}}_{s,r}\hat{\boldsymbol{\alpha}}_{s,r}^H+\sigma_n^2\textbf{R}_{s,r}^{-1}\right)^{-1}\sqrt{M}\hat{\boldsymbol{\alpha}}_{s,r}\right) y_{s,r}^*[m], \hspace{20mm} \mathrm{where} \quad \hat{\boldsymbol{\alpha}}_{s,r}= \hat{\boldsymbol{\alpha}}_{s,r}^{(i-1)}\label{c_r}\\
   & \label{eq:alpha_r}\hat{\boldsymbol{\alpha}}_{s,r}^{(i)}\!=\! 
    \Big(\!\!M\! \!\sum_{m=1}^{\tau_s} \!\hat{\textbf{c}}_{s,r}[m]\hat{\textbf{c}}_{s,r}^H[m]\!+\!\sigma_n^2 (\textbf{R}_{s,r}^\textrm{rcs})^{-1}\!\Big)^{\!\!-1} \!\!\Big (\sqrt{M}\!\sum_{m=1}^{\tau_s}\!\hat{\textbf{c}}_{s,r}[m]{y}_{s,r}[m]\!\Big)\!, \quad\quad \mathrm{where} \quad \hat{\textbf{c}}_{s,r}= \hat{\textbf{c}}_{s,r}^{(i)}\\
&\label{eq:fpartlyinformed}
    T^{\rm PIS,(i)}_{s,r}\! =\!-\hat{\boldsymbol{\alpha}}_{s,r}^H\!\Big(\!\!M\!\!\sum_{m=1}^{\tau_s} \!\hat{\textbf{c}}_{s,r}[m]\hat{\textbf{c}}_{s,r}^H[m]\!+\!\sigma_n^2 (\textbf{R}_{s,r}^\textrm{rcs})^{-1}\!\!\Big)\!\hat{\boldsymbol{\alpha}}_{s,r}\!+\!2\Re\!\Big(\!\sqrt{M}\hat{\boldsymbol{\alpha}}_{s,r}^H\!\Big(\!\sum_{m=1}^{\tau_s} \!\hat{\textbf{c}}_{s,r}[m]
   y_{s,r}[m]\!\Big)\!\Big)\!-\!\sigma_n^2\!\sum_{m=1}^{\tau_s}\!\hat{\textbf{c}}_{s,r}^H[m]\textbf{R}_{s,r}^{-1}\!\hat{\textbf{c}}_{s,r}[m].
\end{align}
\hrulefill \vspace{-5mm}
\end{figure*}

\vspace{-2mm}
\section{Network Resources and Power Modeling}\label{sec:resource_model}
In this section, we develop a cross-layer model for computational complexity, fronthaul requirements, and power consumption in the considered CF-mMIMO ISAC system. The model captures the interplay between radio transmission, fronthaul signaling, and both radio and cloud processing, where decisions in one domain directly affect resource usage and power consumption in the others. 

To facilitate resource modeling and optimization, we introduce binary variables that describe AP operation modes and association decisions. Let $z_l,\overline{z}_l\in\{0,1\}$ indicate whether AP $l$ operates as a TX-AP or RX-AP, respectively, with $z_l+\overline{z}_l\leq1$. Moreover, $\eta_{k,l},\zeta_{s,l},\xi_{s,l}\in\{0,1\}$ indicate the association of AP $l$ with UE $k$ for transmission, SSA $s$ for sensing transmission, and SSA $s$ for sensing reception, respectively.

\vspace{-3.5mm}

\subsection{Fronthaul Data Rate Requirements}

We characterize the communication and sensing fronthaul requirements under functional split Option~7.2, where precoding and combining are performed at the APs. The fronthaul load depends on the AP mode: TX-APs receive communication and sensing data from the cloud, whereas RX-APs send local test statistics to the cloud and receive the side information required for their computation.

For a TX-AP, the cloud forwards the data symbols to the AP, while the AP sends its precoding vectors to the cloud. The cloud uses these vectors to construct $\{\mathbf{c}_{s,r}[m]\}$ in FIS or their covariance matrices in PIS. The fronthaul rate is
\begin{align}\label{eq:R_tx}
    R_l^{\rm tx}
    \!= \!\frac{2N_\textrm{bits}N_\textrm{used}}{\tau_cT_s}
    \Big(\!\!(\tau_d\!+ \!M\!)\!\sum_{k=1}^{K}\!\eta_{k,l} \!+\! (\!\tau_s\!+ \!M)\!\sum_{s=1}^{S}\!\zeta_{s,l} \!\Big)\!,
\end{align}
where $N_{\rm bits}$, $N_{\rm used}$, and $T_s$ are the number of quantization bits, the number of used subcarriers, and the OFDM symbol duration, respectively.

For an RX-AP, i.e., $\overline{z}_l=1$, the fronthaul load comprises: i) the local test statistic sent to the cloud for each assigned SSA and ii) the side information sent from the cloud to compute this statistic. Each RX-AP is assumed to know its own location, the SSA locations, and the RCS correlation matrices. Their signaling overhead is therefore neglected. However, the cloud must notify each RX-AP of its assigned SSAs, requiring one real scalar per assignment.

In FIS, each RX-AP requires
$\mathbf{c}_{s,r}[m]\in\mathbb{C}^{L_{\rm tx}}$ for all $\tau_s$ sensing symbols, corresponding to $2\tau_sL_{\rm tx}$ real scalars per assigned SSA where $L_{\rm tx}=\sum_{l=1}^Lz_l$. In PIS, it instead requires the covariance matrix
$\mathbf{R}_{s,r}\in\mathbb{C}^{L_{\rm tx}\times L_{\rm tx}}$, which contains $L_{\rm tx}^{2}$ real scalars given its Hermitian symmetric nature, and is assumed to be updated once per coherence block. Hence,
\begin{align}\label{eq:R_rx}
R_l^{\rm rx}
=&
\frac{N_{\rm bits}N_{\rm used}}{\tau_cT_s}
\sum_{s=1}^{S}\xi_{s,l}\nonumber\\[-2mm]
&\times
\Big[
2+2\mathbb{I}_{\rm FIS}\tau_s\sum_{j=1}^{L}z_j
+(1-\mathbb{I}_{\rm FIS})
\Big(\sum_{j=1}^{L}z_j\Big)^2
\Big],
\end{align} where $\mathbb{I}_{\rm FIS}\in\{0,1\}$ is an indicator variable that equals one when FIS is employed and zero otherwise. Here, the first term corresponds to the SSA assignment and the local test statistic. Thus, the RX-AP fronthaul load scales linearly with the number of assigned SSAs. It further scales linearly with $L_{\rm tx}$ and $\tau_s$ in FIS, but quadratically with $L_{\rm tx}$ in PIS.

Since $R$ RX-APs are assigned to each SSA, so that
$\sum_{l=1}^{L}\sum_{s=1}^{S}\xi_{s,l}=RS$, the total fronthaul requirement is
\begin{align}
R_{\rm tot}^{\rm front}
={}&
\kappa_{\rm f0}
\sum_{l=1}^{L}\sum_{k=1}^{K}\eta_{k,l}
+
\kappa_{\rm f1}
\sum_{l=1}^{L}\sum_{s=1}^{S}\zeta_{s,l}
+
\kappa_{\rm f2}
\sum_{l=1}^{L}\sum_{s=1}^{S}\xi_{s,l}
\nonumber\\
&+
\kappa_{\rm f3}\sum_{l=1}^{L}z_l
+
\kappa_{\rm f4}
\Big(\sum_{l=1}^{L}z_l\Big)^2,
\label{eq:total_fronthaul_compact}
\end{align}
where
\begin{align}
\kappa_{\rm f0}
&=\frac{2N_{\rm bits}N_{\rm used}}
{\tau_cT_s}(\tau_d+M),
\quad
\kappa_{\rm f1}
=\frac{2N_{\rm bits}N_{\rm used}}
{\tau_cT_s}(\tau_s+M),
\nonumber\\
\kappa_{\rm f2}
&=\frac{2N_{\rm bits}N_{\rm used}}
{\tau_cT_s},
 \hspace{17mm}
\kappa_{\rm f3}
=\frac{2N_{\rm bits}N_{\rm used}}
{\tau_cT_s}\,
\mathbb I_{\rm FIS}\tau_sRS,
\nonumber\\
\kappa_{\rm f4}
&=\frac{N_{\rm bits}N_{\rm used}}
{\tau_cT_s}\,
(1-\mathbb I_{\rm FIS})RS.
\end{align}
\vspace{-5mm}
\subsection{Radio and Cloud GOPS Analysis}\label{sec:GOPS}\vspace{-1.5mm}
We quantify the processing load in giga operations per second (GOPS), considering real multiplications and divisions required for communication and sensing. Following \cite{desset2016massive,demir2023cell}, each complex multiplication is counted as eight operations: four real multiplications, doubled to account for memory-access overhead. Accordingly, multiplying matrices of dimensions $a\times b$ and $b\times c$ requires $abc$ complex multiplications.

The processing load of TX-AP $l$, i.e., $z_l=1$, is
\begin{align}
C_{\textrm{AP},l}^\textrm{tx}= C_{\textrm{filter}}+ C_{\textrm{DFT}}+  C_{\textrm{ch-es},l}+C_{\textrm{prec},l}^\textrm{comm}+\sum_{s=1}^{S} \zeta_{s,l} C_{\textrm{prec}}^\textrm{sens},\nonumber
\end{align}
where $C_{\textrm{filter}}=40Mf_s/10^9$ accounts for baseband filtering at sampling frequency $f_s$, while
$C_{\textrm{DFT}}
=8MN_{\textrm{DFT}}\log_2(N_{\textrm{DFT}})/(T_s10^9)$
accounts for the inverse DFT of the downlink signals \cite{demir2023cell}. Here, $N_{\textrm{DFT}}$ and $T_s$ denote the DFT size and OFDM symbol duration, respectively. The communication channel-estimation and LP-MMSE precoding loads are, respectively, given by \cite{demir2023cell,behdad2026energy}
\begin{align}
    &C_{\textrm{ch-es},l}=\frac{N_\textrm{used}}{T_\textrm{s}\!\tau_c 10^9}\!\Big(\!8M\tau_p^2 \!+\!8M^2\tau_p\!+\!8M^2 \!\sum_{k=1}^{K}\!\eta_{k,l}\!\Big),\\
    &C_{\textrm{prec},l}^\textrm{comm}=\frac{N_\textrm{used}}{T_\textrm{s}\tau_c 10^9}\!\Big[4\tau_d M\sum_{k=1}^{K}\!\eta_{k,l}\!+\!8(\tau_d+1)M\!\sum_{k=1}^{K}\!\eta_{k,l}\nonumber\\
    &\hspace{5mm}+\Big(\!4(M^2+M)\!\tau_p\!+\!M^2\sum_{k=1}^{K}\eta_{k,l}\!+ \!\frac{8(M^3-M)}{3} \!\Big)\!\Big].
\end{align}
The sensing-precoding load is $C_{\textrm{prec}}^\textrm{sens}=\frac{N_\textrm{used}}{T_\textrm{s}\tau_c 10^9} \left(12\tau_s M\right)$ which accounts for applying the MRT sensing precoder and power coefficients.
Collecting the association-dependent terms, the TX-AP processing load becomes
\begin{align}
    C_{\textrm{AP},l}^\textrm{tx}= \kappa_{\rm t1}\sum_{k=1}^{K}\!\eta_{k,l} +  C_{\textrm{prec}}^\textrm{sens}\sum_{s=1}^S \zeta_{s,l} + \kappa_{\rm t0},
    \end{align}
where $\kappa_{t1}=\frac{N_\textrm{used}}{T_\textrm{s}\tau_c 10^9}(9M^2 +12\tau_d M + 8M)$ and $\kappa_{t0}=C_{\textrm{filter}}+ C_{\textrm{DFT}}+ \frac{N_\textrm{used}}{T_\textrm{s}\tau_c 10^9} \left(8M\tau_p^2 + 12M^2\tau_p+ 4M \tau_p +\frac{8M^3-8M}{3} \right)$. 

For an RX-AP, i.e., $\overline{z}_l=1$, the processing load comprises baseband filtering, DFT processing, sensing combining, and computation of the local test statistics:
\begin{align}
C_{\textrm{AP},l}^\textrm{rx}= C_{\textrm{filter}}+ C_{\textrm{DFT}}+ \sum_{s=1}^{S} \xi_{s,l}\left(C_{\textrm{comb}}^\textrm{sens}+C_{\textrm{d}}\right),\label{eq:gops_rx_general}
\end{align}
Here, $C_{\textrm{DFT}}$ accounts for the DFT of the received sensing signals, while $C_{\textrm{comb}}^\textrm{sens}=\frac{N_\textrm{used}}{T_\textrm{s}\tau_c 10^9} \left(8\tau_s M\right)$ is GOPS for multiplying the combining vector per assigned SSA. $C_{\textrm{d}}$ is the required GOPS to obtain local test statistics, given as 
\begin{align}
C_{\textrm{d}} = \mathbb{I}_{\rm FIS}C_{\textrm{maprt}}^{\textrm{FIS}} + (1-\mathbb{I}_{\rm FIS}) n_{\textrm{iter}} C_{\textrm{maprt}}^{\textrm{PIS,(i)}}
\end{align}
where $n_{\textrm{iter}}$ is the number of iterations required by the PIS detector, and $\mathbb I_{\rm FIS}=1$ for FIS and $0$ for PIS.

\begin{proposition}\label{prop:gops}
The processing load for computing one local sensing statistic at an RX-AP is
\begin{align}
C_{\textrm{d}}
=&\kappa_{\rm r3}\Big(\sum_{j=1}^{L} z_j\Big)^3
+\kappa_{\rm r2}\Big(\sum_{j=1}^{L} z_j\Big)^2
+\kappa_{\rm r1}\sum_{j=1}^{L} z_j,
\label{eq:C_d}
\end{align}
where 
\begin{align}
&\kappa_{\rm r1}\!=\!\phi\Big(\! \mathbb{I}_{\rm FIS} \!(8\tau_s \! -\!\frac{14}{3}\!)\!+ \!n_{\textrm{iter}}\!\Big(\!1- \!\mathbb{I}_{\rm FIS}\!\Big)\! (\!18\tau_s\!+\!\frac{14}{3}\!)\! \Big)\!,\\
&\kappa_{\rm r2}=\phi \Big( \mathbb{I}_{\rm FIS}(4\tau_s + 8)+ n_{\textrm{iter}} \left(1- \mathbb{I}_{\rm FIS}\right) (12 \tau_s + 24)\Big),\\
&\kappa_{\rm r3}= \phi \Big( \mathbb{I}_{\rm FIS} \frac{8}{3}+n_{\textrm{iter}} \left(1- \mathbb{I}_{\rm FIS}\right)\frac{16}{3}\Big).
\end{align}
with $\mathbb I_{\rm FIS}=1$ for FIS and $0$ for PIS, and $\phi\triangleq\frac{N_{\rm used}}{T_s\tau_c10^9}$.
Consequently, the total RX-AP processing load in
\eqref{eq:gops_rx_general} becomes

\begin{align}
    C_{\textrm{AP},l}^\textrm{rx}&= C_{\textrm{comb}}^\textrm{sens}\sum_{s=1}^{S} \xi_{s,l} + \kappa_{\rm r3}\sum_{s=1}^{S} \xi_{s,l}\Big(\sum_{j=1}^{L} z_j\Big)^3 \nonumber\\
    &\hspace{-5mm}+\kappa_{\rm r2}\sum_{s=1}^{S} \xi_{s,l}\Big(\sum_{j=1}^{L} z_j\Big)^2 + \kappa_{\rm r1}\sum_{s=1}^{S} \xi_{s,l}\sum_{j=1}^{L} z_j +\kappa_{\rm r0}, \label{eq:gops_rx_compact}
\end{align}
where $\kappa_{\rm r0}=C_{\textrm{filter}}+ C_{\textrm{DFT}}$.
\end{proposition}

\begin{proof}
The result follows by counting the real multiplications and divisions of each detector and grouping the terms by the number of active TX-APs. See Appendix~\ref{app:gops} for details.
\end{proof}

The cloud processing load comprises communication processing, generation of the sensing side information required by the RX-APs, and aggregation of their local test statistics.

\begin{proposition}\label{prop:cloud_complexity}
The total cloud processing load is
\begin{align}
&C_{\rm cloud}
\!=\!F_{\rm fixed}+\kappa_{\rm c1}\sum_{l=1}^{L}z_l
\!+\!\kappa_{\rm c2}
\sum_{l=1}^{L}\sum_{k=1}^{K}\eta_{k,l}
\!+\!\kappa_{\rm c3}
\sum_{l=1}^{L}\sum_{s=1}^{S}\zeta_{s,l}\nonumber\\[-1mm]
&\hspace{5mm}+\!\kappa_{\rm c4}
\Big[
\sum_{k=1}^{K}
\Big(\sum_{l=1}^{L}\eta_{k,l}\Big)^2\!
+\!\sum_{s=1}^{S}
\Big(\sum_{l=1}^{L}\zeta_{s,l}\Big)^2
\Big]
+\kappa_{\rm c5}
\Big(\sum_{l=1}^{L}z_l\Big)^2\nonumber\\[-1mm]
&\hspace{5mm}
+\phi\sum_{l=1}^{L}\sum_{s=1}^{S}\xi_{s,l},
\label{eq:C_cloud_total}
\end{align}
where 
\begin{align}
\kappa_{\rm c1}
=&C_{\rm other\text{-}tx}
+\phi\mathbb I_{\rm FIS}
\left(8\tau_sMS+4\tau_sRS\right),
\nonumber\\
\kappa_{\rm c2}
=&C_{\rm other\text{-}ue}
+\phi\Big[
\mathbb I_{\rm FIS}(8\tau_sM+4)
+(1-\mathbb I_{\rm FIS})S(8M+4)
\Big],
\nonumber\\
\kappa_{\rm c3}
=&\phi\Big[
\mathbb I_{\rm FIS}(8\tau_sM+4)
+(1-\mathbb I_{\rm FIS})S(8M+4)
\Big],
\nonumber\\
\kappa_{\rm c4}
=&4\phi S(1-\mathbb I_{\rm FIS}),
\qquad 
\kappa_{\rm c5}
=2\phi RS(1-\mathbb I_{\rm FIS}).
\end{align}
\end{proposition}
\begin{proof}
See Appendix~\ref{app:cloud_complexity}.
\end{proof}

\vspace{-5mm}
\subsection{Power Consumption Model}\vspace{-1mm}
The total network power consumption comprises the radio-site, fronthaul, and cloud contributions:
\begin{align}
    P_{\text{total}} = P_\textrm{radio}+P_{\text{front}} +P_{\text{cloud}}. 
\end{align}
The radio-site power is $P_{\rm radio}=\sum_{l=1}^{L}P_{{\rm AP},l}$, where
\begin{align}
    &P_{\text{AP},l} = z_l \Big( P_{\text{AP,0}}^{\text{tx}} + \Delta^{\text{tr}} \sum_{k=1}^{K}p_{k,l} \! +  \!\Delta^{\text{tr}}\sum_{s=1}^{S}  q_{s,l} \!+  \!\frac{P_{\rm AP,0}^{\rm proc}}{\sigma_{\rm cool}^{\rm AP}}\nonumber\\
    &+ \!\frac{\Delta_\textrm{AP}^\textrm{proc}C_{\textrm{AP},l}^\textrm{tx}}{\sigma_{\rm cool}^{\rm AP}C_\textrm{AP,max}}\!\Big)  \!\!+  \!\!\overline{z}_l\! \Big( \!P_{\text{AP,0}}^{\text{rx}}\!+ \! \!\frac{P_{\rm AP,0}^{\rm proc}}{\sigma_{\rm cool}^{\rm AP}} \!+  \!\frac{\Delta_\textrm{AP}^\textrm{proc}C_{\textrm{AP},l}^\textrm{rx}}{\sigma_{\rm cool}^{\rm AP}C_\textrm{AP,max}} \!\Big), \!
\end{align}
Here, $P_{{\rm AP},0}^{\rm tx}=6.8M$ and
$P_{{\rm AP},0}^{\rm proc}$ denote the load-independent antenna and idle-GPP power consumption, respectively \cite{demir2022cell}.

The fronthaul power consumption is
$P_{\rm front}
=
P_{\rm ONU}
\sum_{l=1}^{L}(z_l+\bar z_l)$,
where $P_{\rm ONU}$ is the power consumption of an optical network unit (ONU).

The cloud power consumption is
\begin{align}
    P_{\rm cloud}&\!=\!P_{\rm fixed} \!+\!\frac{1}{\sigma_{\rm cool}}\!\Big(\!P_{\text{OLT}}\,\!W \! +\!P_{\rm cloud,0}^{\rm proc}N_{\rm GPP}\!+\!
    \frac{\Delta_{\rm cloud}^{\rm proc}C_{\rm cloud}}{C_{\rm max}} \!\Big),
\end{align}
where $P_{\rm fixed}$ is the fixed cloud power, $P_{\rm OLT}$ is the power per optical line terminal (OLT) module, and
$P_{{\rm cloud},0}^{\rm proc}$ is the idle processing power per GPP. $W$ and $N_{\rm GPP}$ denote the numbers of line cards (LCs) and GPPs, respectively. Assuming one line card per GPP gives $N_{\rm GPP}=W$. The parameters $\sigma_{\rm cool}$, $\Delta_{\rm cloud}^{\rm proc}$, and $C_{\max}$ denote the cooling efficiency, load-dependent processing-power slope, and processing capacity of each GPP, respectively. The cloud load $C_{\rm cloud}$, derived in Section~\ref{sec:GOPS}, includes both communication and sensing processing.
Collecting the constant and load-dependent terms yields the compact power model, expressed in terms of the power-allocation, AP-mode, association, and cloud-resource variables, given as in 
\eqref{eq:P_total},
\begin{figure*}
\begin{align}
   &P_{\rm t}\!=\!a_0+ a_1\! \sum_{l=1}^L\! z_l\! + \!a_2\!\sum_{l=1}^L\! \overline{z}_l\!+\! \Delta^{\rm tr} \!\sum_{l=1}^L\!\sum_{k=1}^K\! p_{k,l}\! +\! \Delta^{\rm tr} \!\sum_{l=1}^L\!\sum_{s=1}^S\! q_{s,l}+\! a_3 \!\sum_{l=1}^L\!\sum_{k=1}^K\! \eta_{k,l} \!+ \!a_4 \!\sum_{l=1}^L\!\sum_{s=1}^S\! \zeta_{s,l} \!+\!a_5\! \sum_{l=1}^L\!\sum_{s=1}^S \!\xi_{s,l}\!+\chi_{\rm w}W +\chi_{\rm AP}\nonumber\\
   &\times\!\Big[\!\kappa_{\rm r3\!}\Big(\!\sum_{j=1}^{L} \!z_j\!\Big)^3
+\!\kappa_{\rm r2}\!\Big(\!\sum_{j=1}^{L}\! z_j\!\Big)^2
+\!\kappa_{\rm r1}\!\sum_{j=1}^{L} \!z_j\!\Big]\!\sum_{l=1}^L\sum_{s=1}^S\! \xi_{s,l}
+
\chi_{\rm c}\phi(1-\mathbb{I}_{\rm FIS}\!)
\Big[\!
4S\!\sum_{k=1}^{K}\Big(\!\sum_{l=1}^{L}\!\eta_{k,l}\!\Big)^2
\!+4S
\!\sum_{s=1}^{S}\!\Big(\!\sum_{l=1}^{L}\!\zeta_{s,l}\!\Big)^2
\!
+\!2RS\!\Big(\!\sum_{l=1}^{L}\!z_l\!\Big)^{2}\!
\Big]\label{eq:P_total}\\
a_1
&=\;
P_{\rm AP,0}^{\rm tx}
+\frac{P_{\rm AP,0}^{\rm proc}}
{\sigma_{\rm cool}^{\rm AP}}
+P_{\rm ONU}
+\chi_{\rm AP}\kappa_{t0}
+\chi_{\rm c}C_{\rm other\text{-}tx}+
\chi_{\rm c}\phi\,
\mathbb{I}_{\rm FIS}
\left(
8\tau_sMS+4\tau_sRS
\right),\label{eq:a1}
\\
a_2
&=\;
P_{\rm AP,0}^{\rm rx}
+\frac{P_{\rm AP,0}^{\rm proc}}
{\sigma_{\rm cool}^{\rm AP}}
+P_{\rm ONU}
+\chi_{\rm AP}\kappa_{r0},
\qquad
a_3
=
\chi_{\rm AP}\kappa_{t1}
+\chi_{\rm c}C_{\rm other\text{-}ue}
+\chi_{\rm c}\phi
\Big[
\mathbb{I}_{\rm FIS}(8\tau_sM+4)
+(1-\mathbb{I}_{\rm FIS})S(8M+4)
\Big],
\label{eq:a3}
\\
a_4&
=\;
\chi_{\rm AP}C_{\rm prec}^{\rm sens}
+\chi_{\rm c}\phi
\Big[
\mathbb{I}_{\rm FIS}(8\tau_sM+4)
+(1-\mathbb{I}_{\rm FIS})S(8M+4)
\Big],
\qquad
a_5=\chi_{\rm AP}C_{\rm comb}^{\rm sens}
+\chi_{\rm c}\phi . \label{eq:a5}
\end{align}
\end{figure*}
where $\chi_{\rm w}= \frac{P_{\rm OLT}+ P_{\rm cloud,0}^{\rm proc}}{\sigma_{\rm cool}}$, $\chi_{\rm AP}
\triangleq
\frac{\Delta_{\rm AP}^{\rm proc}}
{\sigma_{\rm cool}^{\rm AP}C_{\rm AP,\max}}$, 
$\chi_{\rm c}
\triangleq
\frac{\Delta_{\rm cloud}^{\rm proc}}
{\sigma_{\rm cool}C_{\max}}$, $a_0=P_{\rm fixed}+\chi_{\rm c}F_{\rm fixed},$ and $a_1-a_5$ are given in \eqref{eq:a1}-\eqref{eq:a5} on the top of the next page. 

\vspace{-2.5mm}
\section{Network Orchestration and Resource Optimization
}\label{sec:optimization}\vspace{-1mm}
We formulate a resource-orchestration problem that minimizes total radio, fronthaul, and cloud processing power consumption, subject to communication and sensing performance requirements, as well as fronthaul and processing capacity limitations. 
The optimization problem is formulated as
\begin{subequations}\label{opt1}
\begin{align}
 \mathcal{P}_0:& \underset{
  \boldsymbol{\rho}, \boldsymbol{q}\geq \textbf{0}, z_l,\overline{z}_l, \eta_{k,l}, \zeta_{s,l}, \xi_{s,l}, W}{\textrm{minimize}} \quad  P_{\rm t}
    \label{obj_func}\\
  & \textrm{subject to} \nonumber\\
   & \hspace{5mm}
    \mathsf{SINR}_k^{\text{comm}}
    \geq \gamma_c,\quad \forall k\label{const:comm_sinr}\\
    &\hspace{5mm}\mathsf{SINR}_{s,r}^{\text{sens}}\geq \gamma_{\rm s}\xi_{s,r}, \quad \forall s, r  \label{const:sensing_sinr}\\
    &\hspace{5mm}\sum_{k=1}^{K}p_{k,l}+\sum_{s=1}^{S}q_{s,l}\leq P_{\rm max}\,z_l,\quad\quad \quad \forall l\label{const:power}\\
    & \hspace{5mm} 0\leq p_{k,l}\leq P_{\text{max}}\,\eta_{k,l}, \quad \forall k, l\label{const:zero_comm_power}\\
    & \hspace{5mm} 0\leq q_{s,l}\leq P_{\text{max}}\,\zeta_{s,l}, \quad \forall s, l \label{const:zero_sens_power}\\
    &\hspace{5mm}\frac{1}{K+S}\!\left(\!\sum_{k=1}^{K}\!\eta_{k,l}\!+\!\sum_{s=1}^{S}\!\zeta_{s,l}\!\right)\!\leq\! z_l \!\leq\!\sum_{k=1}^{K}\!\eta_{k,l}\!+\!\sum_{s=1}^{S}\!\zeta_{s,l}, \hspace{1mm}\forall l\label{const:TXselection}\\
    & \hspace{5mm}\frac{1}{S}\sum_{s=1}^{S}\xi_{s,l} \leq \overline{z}_l \leq\sum_{s=1}^{S}\xi_{s,l}, \quad\forall l\label{const:RXselection}\\
     & \hspace{5mm}z_l+\overline{z}_{l} \leq 1, \quad \forall l\label{const:APmode}\\
    & \hspace{5mm}\zeta_{s,l}+\xi_{s,l} \leq 1, \quad \forall s,l\label{const:APselection}\\
    &\hspace{5mm} \sum_{l=1}^L \xi_{s,l}= R, \quad \forall s \label{const:sumRX per SSA}\\
    & \hspace{5mm} C_{\textrm{AP},l}^\textrm{tx}z_l\leq C_{\rm AP,max}, \quad \forall l \label{const:processing_TX}\\
     & \hspace{5mm} C_{\textrm{AP},l}^\textrm{rx}\bar{z}_l\leq  C_{\rm AP,max}, \quad \forall l\label{const:processing_RX}\\
      & \hspace{5mm} C_{\textrm{cloud}}\leq W C_{\rm max}, \label{const:processing_cloud}\\
     & \hspace{5mm}  R^{\rm front}_{\rm tot} \leq R^{\rm f}_{\rm max}W, \label{const:front1}\\
     & \hspace{5mm} W\in \{1,2, ..., W_{\rm max}\}\label{const:front2}\\
   & \hspace{5mm} z_l, \overline{z}_l, \eta_{k,l}, \zeta_{s,l}, \xi_{s,l}\in \{0,1\}\label{const:binary},
\end{align}
\end{subequations}
Here, $\mathcal{P}_0$ jointly determines the AP operation modes, UE and sensing-area associations, transmit power coefficients, and the number of active LCs and GPPs.
Constraint \eqref{const:comm_sinr} guarantees a minimum communication SINR $\gamma_c$ for each UE, while \eqref{const:sensing_sinr} imposes a minimum sensing SINR $\gamma_s$ for each selected SSA--RX-AP pair. Constraints \eqref{const:power}--\eqref{const:zero_sens_power} enforce the AP power budgets and ensure that power is allocated only to associated UEs and SSAs.
Constraints \eqref{const:TXselection} and \eqref{const:RXselection} link the AP modes to their associations: an AP operates as a TX-AP or RX-AP only if it has at least one corresponding assignment. Constraint \eqref{const:APmode} allows each AP to operate in TX, RX, or idle mode, while \eqref{const:APselection} prevents an AP from simultaneously transmitting toward and receiving echoes from the same SSA, thereby avoiding mono-static full-duplex self-interference. Constraint \eqref{const:sumRX per SSA} assigns exactly $R$ RX-APs to each SSA.
Constraints \eqref{const:processing_TX} and \eqref{const:processing_RX} limit the TX- and RX-processing loads at each AP, respectively, while \eqref{const:processing_cloud} ensures that the cloud load does not exceed the aggregate capacity of the active GPPs. Constraint \eqref{const:front1} limits the total fronthaul data rate to the aggregate capacity of the active LCs, where $R_{\rm max}^{\rm f}$ is the capacity of each LC. Finally, \eqref{const:front2} and \eqref{const:binary} impose the integer and binary constraints, respectively.

Problem $\mathcal{P}_0$ is evaluated under three configurations that differ in their coordination and resource-allocation flexibility. Under \emph{local coordination}, a fixed amount of fronthaul and processing resources is statically assigned to groups of APs. Consequently, an LC and its associated GPP can be deactivated only when all connected APs are inactive. Under \emph{full coordination}, the total available resources remain fixed, but their association with the active APs can be dynamically reconfigured, enabling more efficient resource sharing.
The proposed \emph{E2E optimization} extends full coordination by dynamically optimizing both the resource associations and the allocated capacity. In particular, it jointly optimizes the radio, fronthaul, and cloud variables, including the number of active LCs and GPPs, according to the network demand.

\vspace{-3.5mm}
\subsection{E2E Optimization}
The proposed E2E framework employs full coordination with flexible resource allocation across the radio, fronthaul, and cloud domains. 
Problem $\mathcal{P}_0$ is a mixed-integer non-convex program due to the communication and sensing SINR constraints in \eqref{const:comm_sinr} and \eqref{const:sensing_sinr}, the coupling between the detector load and the AP-mode and RX-association variables, and the coexistence of continuous, binary, and integer variables. To obtain a tractable solution, we reformulate the communication SINR constraints in SOC form \cite[Sec.~7.1.2]{cell-free-book}, handle the conditional sensing constraints using a Big-$M$ formulation \cite{vielma2015mixed} and FPP-SCA \cite{mehanna2014feasible}, and relax the binary and integer variables. A feasible discrete solution is subsequently recovered through a structured projection.

We introduce the relaxed variables
$\tilde z_l,\hat z_l,\tilde\eta_{k,l},
\tilde\zeta_{s,l},\tilde\xi_{s,l}\in[0,1]$
and replace the integer variable $W$ with $\tilde W\in[1,W_{\max}]$. At iteration $c$, the binary variables recovered in the preceding projection are fixed.
To minimize the error, we add a mean-square error (MSE) penalty to the objective function for all binary variables as in \eqref{eq:MSE}, on the top of the next page.
\begin{figure*}
\begin{align}
&\mathsf{MSE}=\lambda_1 \sum_{l=1}^L\Big(z_l-\sqrt{\tilde{z}_l+\epsilon_a}\Big)^2+\lambda_2 \sum_{l=1}^L\Big(\overline{z}_l-\sqrt{\hat{z}_l+\epsilon_a}\Big)^2+ \!\lambda_3 \!\sum_{l=1}^L\!\sum_{k=1}^K\Big(\!\eta_{k,l}\!-\!\sqrt{\tilde{\!\eta}_{k,l}\!+\!\epsilon_a}\!\Big)\!^2\!+\!\lambda_4\! \sum_{l=1}^L\!\sum_{s=1}^S\!\Big(\!\zeta_{s,l}\!-\!\sqrt{\!\tilde{\zeta}_{s,l}\!+\!\epsilon_a}\!\Big)\!^2\nonumber\\[-2mm]
    &\hspace{10mm}+ \lambda_5 \sum_{l=1}^L\sum_{s=1}^S\Big(\xi_{s,l}-\sqrt{\tilde{\xi}_{s,l}+\epsilon_a}\Big)^2. \label{eq:MSE}
\end{align}
\hrulefill \vspace{-5mm}
\end{figure*}
where $\epsilon_a$ is a small positive constant introduced to prevent $z_l=0$ when the power coefficients and $\tilde{z}_l$ are small. The square-root transformation biases the relaxed variable toward larger values in $[0,1]$, thereby promoting solutions closer to binary decisions. 

The communication SINR constraint, for $k=1,\ldots,K$, is written in SOC form 
\begin{align}
     &\left\lVert \begin{bmatrix}  \boldsymbol{\rho}_1^T\textbf{B}_{k1}^{\frac{1}{2}} \ldots  \boldsymbol{\rho}_K^T\textbf{B}_{kK}^{\frac{1}{2}} \ 
     \textbf{q}_1^T\textbf{C}_{k1}^{\frac{1}{2}} \ldots  \textbf{q}_S^T\textbf{C}_{kS}^{\frac{1}{2}} \,
    \sigma_n
    \end{bmatrix}\right\rVert \leq \frac{\textbf{a}_k^T \boldsymbol{\rho}_k}{\sqrt{\gamma_c}}. 
    \label{const:comm_SINR_SOC}
\end{align}

To handle the binary variable $\xi_{s,r}$ in \eqref{const:sensing_sinr}, we first reformulate the sensing SINR constraint using the Big-$M$ formulation \cite{vielma2015mixed} and rewrite it as the following SOC form:
\begin{align}
\sqrt{\gamma_s}
\left\lVert \begin{bmatrix}\boldsymbol{\mathsf{B}}_{s,r}^{\frac{1}{2}}\boldsymbol{\rho}  \\
    \sigma_n\sqrt{\tau_s}
    \end{bmatrix}\right\rVert\leq\left\lVert\boldsymbol{\mathsf{A}}_{s,r}^{\frac{1}{2}}\boldsymbol{\rho}\right\rVert+M_{s,r}(1-\xi_{s,r}),\label{const:bigM_sensing}
\end{align}
where $M_{s,r}>0$ is a sufficiently large constant, selected as 
$M_{s,r} =
(1+\epsilon_M)
\sqrt{\gamma_s\left(
P_{\max}\sum_{l=1}^{L}\lambda_{\max}
(\boldsymbol{\mathsf{B}}_{s,r}^{(l)})
+\tau_s\sigma_n^2
\right)}$,
where $\lambda_{\max}(\cdot)$ denotes the largest eigenvalue,
 $\boldsymbol{\mathsf{B}}_{s,r}^{(l)}$ is the $l$-th AP block of
$\boldsymbol{\mathsf{B}}_{s,r}$ and $\epsilon_M>0$ is a small safety margin. 
When $\xi_{s,r}=1$, \eqref{const:bigM_sensing} reduces to the original sensing SINR requirement, whereas for $\xi_{s,r}=0$, the Big-$M$ term relaxes the constraint.

Constraint \eqref{const:bigM_sensing} remains nonconvex due to the convex term
$\left\lVert\boldsymbol{\mathsf{A}}_{s,r}^{\frac{1}{2}}\boldsymbol{\rho}\right\rVert$
on its right-hand side. We therefore apply the FPP-SCA approach and replace this term by its first-order affine lower bound around the previous iterate $\boldsymbol{\rho}^{(i)}$. %
Hence, at iteration $i+1$, the sensing SINR constraint, taking into account the relaxed variables, is approximated as
\begin{align}
&\left\lVert\boldsymbol{\mathsf{A}}_{s,r}^{\frac{1}{2}}\boldsymbol{\rho}^{(i)}\right\rVert
+\Re\left\{\frac{\Big(\boldsymbol{\mathsf{A}}_{s,r}^{\frac{1}{2}}\boldsymbol{\rho}^{(i)}\Big)^{H}}{\left\lVert\boldsymbol{\mathsf{A}}_{s,r}^{\frac{1}{2}}\boldsymbol{\rho}^{(i)}\right\rVert}\boldsymbol{\mathsf{A}}_{s,r}^{\frac{1}{2}}
\Big(\boldsymbol{\rho}-\boldsymbol{\rho}^{(i)}\Big)\!\right\}\!\nonumber\\
&\hspace{10mm}+M_{s,r}(1-\tilde{\xi}_{s,r})\!+\!\Xi_{s,r}\geq \sqrt{\gamma_s}\left\lVert\begin{bmatrix}  \boldsymbol{\mathsf{B}}_{s,r}^{\frac{1}{2}}\boldsymbol{\rho}  \\
    \sigma_n\sqrt{\tau_s}
    \end{bmatrix}
\right\rVert,
\label{const:ccp_soc_sensing}
\end{align}
where $\Xi_{s,r}\!\geq\! 0$ is a nonnegative slack variable introduced to preserve feasibility during the initial FPP-SCA iterations. The objective is augmented with the penalty term $\lambda_0\!\sum_{s=1}^{S}\!\sum_{r=1}^{L}\!\Xi_{s,r}$, where $\lambda_0$ is chosen sufficiently large to discourage constraint violations and drive the solution toward feasibility of the original problem.

Using the association--mode coupling in \eqref{const:TXselection}, the TX-AP processing capacity constraint is equivalently written as
\begin{align}
\kappa_{\rm t1}\sum_{k=1}^{K}\tilde{\eta}_{k,l}
+C_{\rm prec}^{\rm sens}\sum_{s=1}^{S}\tilde{\zeta}_{s,l}
+\kappa_{\rm t0}\tilde{z}_l
\leq C_{\rm AP,max},\quad \forall l.
\label{const:processing_TX_p1}
\end{align}
Similarly, using the RX-association constraint in
\eqref{const:RXselection}, the RX-AP processing capacity constraint becomes
\begin{align}
\left(C_{\rm comb}^{\rm sens}+C_{\rm d}\right)
\sum_{s=1}^{S}\tilde{\xi}_{s,l}
+\kappa_{\rm r0}\hat{z}_l
\leq C_{\rm AP,max},\quad \forall l.
\end{align}
This constraint is not convex because $C_{\rm d}$ depends on the number of active TX-APs and is multiplied by the RX-association term. To address this nonconvexity, we  define
$\tilde{u}_l
\triangleq
\sum_{s=1}^{S}\tilde{\xi}_{s,l}$ and
$\tilde{t}
\triangleq
\sum_{j=1}^{L}\tilde{z}_{j}$,
and introduce the auxiliary variable $\tilde{C}_{\rm d}\geq0$ satisfying
\begin{align}
\tilde{C}_{\rm d}
\geq
\kappa_{\rm r3}\tilde{t}^{3}
+\kappa_{\rm r2}\tilde{t}^{2}
+\kappa_{\rm r1}\tilde{t}.
\label{const:Cd_upper}
\end{align}
The bilinear term can be expressed in difference-of-convex form as
$ 2\tilde{C}_{\rm d}\tilde{u}_l
=
(\tilde{C}_{\rm d}+\tilde{u}_l)^2
-\tilde{C}_{\rm d}^{\,2}
-\tilde{u}_l^2$.
At iteration $i$ of the SCA procedure, the convex terms
$\tilde{C}_{\rm d}^{\,2}$ and $\tilde{u}_l^2$ are replaced by their first-order affine lower bounds at
$\tilde{C}_{\rm d}^{(i)}$ and $\tilde{u}_l^{(i)}$, respectively. This yields the following convex inner approximation of the RX-AP processing capacity constraint:
\begin{align}
&2C_{\rm comb}^{\rm sens}\tilde{u}_l
+2\kappa_{\rm r0}\hat{z}_l
+(\tilde{C}_{\rm d}+\tilde{u}_l)^2
\nonumber\\
&\leq
2C_{\rm AP,max}
+2\tilde{C}_{\rm d}^{(i)}\tilde{C}_{\rm d}
-\Big(\tilde{C}_{\rm d}^{(i)}\Big)^2\!
+\!2\tilde{u}_l^{(i)}\tilde{u}_l\!
-\!\Big(\tilde{u}_l^{(i)}\Big)^2\!,
\quad \forall l.
\label{const:processing_RX_p1}
\end{align}
Constraints \eqref{const:Cd_upper} and \eqref{const:processing_RX_p1} are updated at each CCP iteration. 

With the relaxed variables, the cloud processing capacity constraint in
\eqref{const:processing_cloud} becomes
\begin{align}
    &\kappa_{\rm c1}\sum_{l=1}^{L}\tilde z_l+\kappa_{\rm c2}\sum_{l=1}^{L}\sum_{k=1}^{K}\tilde\eta_{k,l}+\kappa_{\rm c3}\sum_{l=1}^{L}\sum_{s=1}^{S}\tilde\zeta_{s,l}+\phi\sum_{l=1}^{L}\sum_{s=1}^{S}\tilde\xi_{s,l}\nonumber\\
      &+\kappa_{\rm c4}\Bigg[\sum_{k=1}^{K}\Big(\sum_{l=1}^{L}\tilde\eta_{k,l}\Big)^2+\sum_{s=1}^{S}\Big(\sum_{l=1}^{L}\tilde\zeta_{s,l}\Big)^2\Bigg]\!+\!\kappa_{\rm c5}\Big(\sum_{l=1}^{L}\tilde z_l\Big)^2\nonumber\\
      &\hspace{5mm}+F_{\rm fixed} \leq C_{\rm max}\tilde W.\label{const:processing_cloud_p1}
\end{align}

Similarly, the relaxed fronthaul-capacity constraint is
\begin{align}
    &\kappa_{\rm f0}\sum_{l=1}^{L}\sum_{k=1}^{K}\tilde\eta_{k,l}+\kappa_{\rm f1}\sum_{l=1}^{L}\sum_{s=1}^{S}\tilde\zeta_{s,l}+\kappa_{\rm f2}\sum_{l=1}^{L}\sum_{s=1}^{S}\tilde\xi_{s,l}\nonumber\\
     &\quad\quad+\kappa_{\rm f3}\sum_{l=1}^{L}\tilde z_l+\kappa_{\rm f4}
\Big(\sum_{l=1}^{L}\tilde z_l\Big)^2
\leq R_{\rm max}^{\rm f}\tilde W.
\label{const:front_p1}
\end{align}

We develop a two-stage iterative algorithm that first solves a relaxed continuous problem with binary penalty and slack terms and then recovers the binary and integer decisions. Since each SSA is assigned exactly $R$ RX-APs, the total number of assignments is
$\sum_{l=1}^{L}\sum_{s=1}^{S}\tilde{\xi}_{s,l}=RS$. Accordingly, the total power consumption under the relaxed variables is \eqref{eq:P_tilda}, on the top of the next page.
\begin{figure*}
\begin{align}
&\tilde{P}_{\rm t}=a_0+a_1 \sum_{l=1}^L \tilde{z}_l + a_2  \sum_{l=1}^L \hat{z}_l+ \Delta^{\rm tr} \left(\left\lVert\boldsymbol{\rho}\right\rVert_F^2+
\left\lVert\boldsymbol q\right\rVert_F^2\right)+ a_3 \sum_{l=1}^L\sum_{k=1}^K \tilde{\eta}_{k,l} + a_4 \sum_{l=1}^L\sum_{s=1}^S \tilde{\zeta}_{s,l} +a_5 \sum_{l=1}^L\sum_{s=1}^S \tilde{\xi}_{s,l}+ \chi_{\rm w}\tilde{W}+ \chi_{\rm AP}RS\nonumber\\
    &\times\Big(\!\kappa_{\rm r3}\Big(\sum_{j=1}^{L} \tilde{z}_j\Big)^3+\kappa_{\rm r2}\Big(\sum_{j=1}^{L} \tilde{z}_j\Big)^2
+\kappa_{\rm r1}\sum_{j=1}^{L} \tilde{z}_j\Big) +\chi_{\rm c}\phi(1-\mathbb{I}_{\rm FIS})\Big[
4S\Big(\sum_{k=1}^{K}\Big(\sum_{l=1}^{L}\tilde{\eta}_{k,l}\Big)^2
+\sum_{s=1}^{S}(\sum_{l=1}^{L}\tilde{\zeta}_{s,l})^2\Big)
+2RS\Big(\sum_{l=1}^{L}\tilde{z}_l\Big)^{2}
\Big]. \label{eq:P_tilda}
\end{align}
\hrulefill \vspace{-5mm}
\end{figure*}

The resulting continuous subproblem is formulated as
\begin{subequations}
\begin{align}
 &\mathcal{P}_1: \underset{
\boldsymbol{\rho},\boldsymbol{q}\geq\mathbf{0},\,
\tilde{\boldsymbol z},\hat{\boldsymbol z},\,
\tilde{\boldsymbol\eta},\tilde{\boldsymbol\zeta},
\tilde{\boldsymbol\xi},\tilde{ W},
\boldsymbol\Xi,\tilde{C}_{\rm d}}{\textrm{minimize}}   \tilde{P}_{\rm t} 
    +\lambda_0 \sum_{r=1}^L\sum_{s=1}^S \Xi_{s,r}+\mathsf{MSE}
    \\
  & \textrm{subject to} \quad \eqref{const:comm_SINR_SOC}, \eqref{const:ccp_soc_sensing}, \eqref{const:processing_TX_p1},\eqref{const:Cd_upper},\eqref{const:processing_RX_p1},\eqref{const:processing_cloud_p1}, \eqref{const:front_p1}\nonumber\\
  &\hspace{5mm}\sum_{k=1}^{K}p_{k,l}+\sum_{s=1}^{S}q_{s,l}\leq p_{\rm max}\,\tilde{z}_l,\quad\quad \quad \forall l\label{const:powerP1}\\
    & \hspace{5mm} 0\leq p_{k,l}\leq P_{\text{max}}\,\tilde{\eta}_{k,l}, \quad \forall k, l\label{const:zero_comm_powerP1}\\
    & \hspace{5mm} 0\leq q_{s,l}\leq P_{\text{max}}\,\tilde{\zeta}_{s,l}, \quad \forall s, l \label{const:zero_sens_powerP1}\\
    &\hspace{5mm}\frac{1}{K+S}\!\Big(\!\sum_{k=1}^{K}\!\tilde{\eta}_{k,l}\!+\!\sum_{s=1}^{S}\!\tilde{\zeta}_{s,l}\!\Big)\! \leq \!\tilde{z}_l\!\leq\!\sum_{k=1}^{K}\!\tilde{\eta}_{k,l}\!+\!\sum_{s=1}^{S}\!\tilde{\zeta}_{s,l},\quad \forall l\label{const:TXselectionP1}\\
    & \hspace{5mm}\frac{1}{S}\sum_{s=1}^{S}\tilde{\xi}_{s,l} \leq \hat{z}_l \leq\sum_{s=1}^{S}\tilde{\xi}_{s,l},\quad \forall l\label{const:RXselectionP1}\\
     & \hspace{5mm}\tilde{z}_l+\hat{z}_{l} \leq 1, \quad \forall l\label{const:APmodeP1}\\
    & \hspace{5mm}\tilde{\zeta}_{s,l}+\tilde{\xi}_{s,l} \leq 1, \quad \forall s,l\label{const:APselectionP1}\\
     &\hspace{5mm} \sum_{l=1}^L \tilde{\xi}_{s,l}= R, \quad \forall s \label{const:sumRX per SSAP1}\\
     & \hspace{5mm} 1\leq \tilde{W}\leq W_{\rm max}\\
   & \hspace{5mm} \tilde{z}_l, \hat{z}_l, \tilde{\eta}_{k,l}, \tilde{\zeta}_{s,l}, \tilde{\xi}_{s,l}\in [0,1]\label{const:binary-relax}.
\end{align}
\end{subequations}

After obtaining the relaxed solution
$(\tilde{\mathbf z},\hat{\mathbf z},
\tilde{\boldsymbol{\eta}},
\tilde{\boldsymbol{\zeta}},
\tilde{\boldsymbol{\xi}},\tilde W)$, the second stage recovers binary decisions through the constrained projection
\begin{align}
    \mathcal{P}_2:\quad
    &\underset{\boldsymbol{z},\bar{\boldsymbol{z}},
    \boldsymbol{\eta},\boldsymbol{\zeta},\boldsymbol{\xi}}
    {\operatorname{minimize}}
    \quad \mathrm{MSE}
    \label{prob:P2_binary}\\
    &\text{subject to}\quad \eqref{const:TXselection}, \eqref{const:RXselection},\eqref{const:APmode},\eqref{const:APselection},\eqref{const:sumRX per SSA}, \eqref{const:binary}.\nonumber
\end{align}
We implement \eqref{prob:P2_binary} using a structured recovery procedure that preserves the coupling between the AP modes and association variables. Define
$\mathcal Q(u)=1$ if $u\geq1/2$ and $\mathcal Q(u)=0$ otherwise. The communication and sensing TX-association variables and TX mode of an AP $l$ are first recovered as $x_{\rm b}^{\star}=\mathcal Q(\sqrt{\tilde{x}^{\star}+\epsilon_a})$, for $x\in\{\eta_{k,l},\zeta_{s,l}, z_l\}$\footnote{The square-root transformation emphasizes small association values during recovery, helping prevent $z_l$ associated with low power coefficients from being rounded to zero and thereby violating the coupling constraints.}. 

To enforce mutually exclusive TX/RX operation, RX-APs are selected only among APs that are not operating in TX mode. The candidate set is therefore
 $\mathcal{L}_{\rm rx}^{\rm cand}
=
\left\{
l\in\mathcal{L}:z_l^{\star}=0
\right\}$.
For each SSA $s$, let $\mathcal R_s^{\star}$ contain the indices of the $R$ candidate APs with the largest relaxed RX-association values:
\begin{align}
\mathcal R_s^{\star}
=
\underset{\substack{\mathcal R\subseteq\mathcal L_{\rm rx}^{\rm cand}\\
|\mathcal R|=R}}
{\operatorname{arg\,max}}
\sum_{l\in\mathcal R}\tilde\xi_{s,l}^{\star}.
\end{align}
The binary RX associations are then recovered as
$\xi_{s,l}^{\star}
= 1$ for $l\in\mathcal R_s^{\star}$ and 0 otherwise. The RX mode follows directly as
$\bar z_l^{\star}
=
\mathbbm{1}\left\{
\sum_{s=1}^{S}\xi_{s,l}^{\star}>0
\right\}$.
Provided that $|\mathcal L_{\rm rx}^{\rm cand}|\geq R$, this structured recovery satisfies the AP-mode and association constraints
\eqref{const:TXselection}--\eqref{const:sumRX per SSA}.

The integer number of LC/GPP pairs is then obtained from the exact recovered cloud processing and fronthaul loads as
\begin{align}\label{eq:w_update}
W =
\max\left\{
1,\left\lceil \frac{C_{\rm cloud}}{C_{\max}}\right\rceil,
\left\lceil \frac{R_{\rm front}}{R_{\rm front}^{\max}}\right\rceil
\right\}.
\end{align}
For the recovered binary associations, the power coefficients are
re-optimized. The resulting solution is then  checked against
all constraints of the original problem $\mathcal{P}_0$.

Algorithm~\ref{algo:E2E} summarizes the proposed E2E procedure, which alternates between relaxed continuous optimization and structured discrete recovery until the NMSE-based convergence criterion is satisfied. Let \(\mathbf{b}\) collect the recovered binary variables \((\mathbf{z},\bar{\mathbf{z}},\boldsymbol{\eta},\boldsymbol{\zeta},\boldsymbol{\xi})\), and let \(\tilde{\mathbf{b}}\) collect their relaxed counterparts. After each recovery step, \(\mathrm{NMSE}=\|\mathbf{b}-\tilde{\mathbf{b}}\|_2^2/\|\mathbf{b}\|_2^2\) is computed. The algorithm is initialized with \(\mathrm{NMSE}=1\) and stops when \(\mathrm{NMSE}<\epsilon\) or the maximum number of iterations is reached.
\begin{algorithm}[t]
\caption{\textbf{E2E Optimization}}
\label{algo:E2E}
\begin{algorithmic}[1]
\STATE \textbf{Input:} System parameters, QoS thresholds, resource limits, penalty coefficients, maximum number of iterations $C_{\rm max}$, and tolerance $\epsilon$.
\STATE \textbf{Initialize:} feasible $(\boldsymbol{\rho}^{(0)},\mathbf{q}^{(0)})$, set $\mathbf{x}^{(0)}=[\boldsymbol{\rho}^{(0)T}\ \mathbf{q}^{(0)T}]^T$, choose relaxed binary/integer variables $(\tilde{\mathbf{z}}^{(0)},\hat{\mathbf{z}}^{(0)},\tilde{\boldsymbol{\eta}}^{(0)},\tilde{\boldsymbol{\zeta}}^{(0)},\tilde{\boldsymbol{\xi}}^{(0)},\tilde{W}^{(0)})$, set iteration counter $c\!\leftarrow\!0$ and $\text{NMSE}\!\leftarrow\!\infty$.
\WHILE{$\mathrm{NMSE} \geq \epsilon$ and $c < C_{\max}$}
\STATE $c\leftarrow c+1$
    \STATE Solve the convexified problem $\mathcal{P}_1$ to update the continuous and relaxed variables.
    \STATE Recover the binary variables from $\mathcal{P}_2$, and project the relaxed integer variable onto the feasible set.
\ENDWHILE
\STATE Re-optimize the power coefficients for the final recovered binary associations and recompute $W$ from \eqref{eq:w_update}.
\STATE Check the feasibility of the recovered solution with respect to $\mathcal{P}_0$.
\STATE \textbf{Return:} $(\boldsymbol{\rho},\mathbf{q},\mathbf{z},\overline{\mathbf{z}},\boldsymbol{\eta},\boldsymbol{\zeta},\boldsymbol{\xi},W)$.\vspace{-1mm}
\end{algorithmic}
\end{algorithm}
\vspace{-5mm}
\subsection{Benchmarks}
We consider two benchmark sets with reduced optimization scopes. In the first set, denoted by \emph{$P_{\rm tx}$ optimization}, the AP modes and UE/SSA associations are determined using a heuristic procedure, after which only the transmit-power coefficients are optimized. In the second set, denoted by \emph{radio optimization}, the transmit-power coefficients, AP modes, and association variables are jointly optimized, while the fronthaul and processing capacity allocated per AP remain fixed. In both cases, $W$ is determined by the prescribed resource-assignment rule rather than jointly optimized.

Each benchmark is evaluated under local and full coordination, resulting in four schemes: i) \emph{$P_{\rm tx}$ opt./local}, ii) \emph{$P_{\rm tx}$ opt./full}, iii) \emph{Radio opt./local}, and iv) \emph{Radio opt./full}.

For these benchmarks, each AP is allocated its maximum required fronthaul and processing capacity. The maximum TX- and RX-AP fronthaul rates, obtained from \eqref{eq:R_tx} and \eqref{eq:R_rx}, are respectively given as
\begin{align}
R_{\max}^{\rm tx}
={}&
\frac{2N_{\rm bits}N_{\rm used}}{\tau_cT_s}
\Big[
(\tau_d+M)K+(\tau_s+M)S
\Big],
\\
R_{\max}^{\rm rx}
={}&
\frac{N_{\rm bits}N_{\rm used}}{\tau_cT_s}S
\Big[
2+2\mathbb I_{\rm FIS}\tau_sL
+(1-\mathbb I_{\rm FIS})L^2
\Big].
\end{align}
Hence, the maximum number of APs supported by each LC is
$N_{\rm w}
=
\left\lfloor
\frac{R_{\max}^{\rm f}}
{\max\left\{
R_{\max}^{\rm tx},
R_{\max}^{\rm rx}
\right\}}
\right\rfloor$
. 

Under local coordination, AP--LC associations are fixed, and the required number of LCs is
$W_{\max}^{\rm local}
=
\left\lceil\frac{L}{N_{\rm w}}\right\rceil$.
An LC can be deactivated only when all APs statically associated with it are inactive. Under full coordination, the resource capacity per AP remains fixed, but active APs can be dynamically reassigned among the LCs. Consequently, the required number of active LCs is
$W^{\rm full}
=
\left\lceil
\frac{L_{\rm tx}+L_{\rm rx}}{N_{\rm w}}
\right\rceil$ where
$L_{\rm tx}+L_{\rm rx}\leq L$.
Unlike these benchmarks, the proposed E2E scheme jointly optimizes resource associations and the allocated fronthaul and processing capacities.

\emph{Heuristic Energy-Unaware Association:}
In the \emph{$P_{\rm tx}$ opt.} benchmarks, the binary AP-mode and association variables are first determined using a heuristic energy-unaware algorithm based on communication and sensing channel gains. The transmit-power coefficients are then optimized for the resulting $z_l$, $\bar z_l$, $\eta_{k,l}$, $\zeta_{s,l}$, and $\xi_{s,l}$.

For each SSA, the APs are ranked according to their sensing channel gains. The highest-ranked AP is selected as the first RX-AP, while the second-highest available AP is selected as the first TX-AP, provided that it has not already been assigned as an RX-AP. This selection promotes a high two-way sensing gain while maintaining mutually exclusive TX/RX operation.

The remaining $R-1$ RX-APs per SSA are selected from the third- to $(R+1)$th-ranked APs, and the additional $T-1$ TX-APs are selected from the remaining available APs. Since an RX-AP may serve multiple SSAs, the number of active RX-APs satisfies $L_{\rm rx}\leq RS$.
For communication, we employ a user-centric association in which each single-antenna UE is jointly served by a subset of ISAC TX-APs whose aggregate channel gain exceeds a prescribed threshold \cite{cell-free-book}. A TX-AP associated with both a UE and an SSA simultaneously transmits downlink data and sensing signals. The complete procedure is provided in \cite[Algorithm~1]{behdad2025detecting}.

\vspace{-3.5mm}
\subsection{Refinement Algorithm}
After obtaining an initial solution, we apply a refinement procedure to further reduce the number of active APs. At each iteration, the power coefficients $\rho_{k,l}$ whose normalized values are below a small threshold $0<\epsilon_p\ll1$ are set to zero, and the active AP set is updated accordingly. The transmit-power coefficients are then reoptimized using the \emph{$P_{\rm tx}$ opt.} algorithm. This procedure is repeated using the optimized powers from the preceding iteration until the problem becomes infeasible, after which the last feasible solution is returned. The threshold $\epsilon_p$ is selected sufficiently small to remove negligible coefficients while preserving feasibility in the first refinement iteration.
\begin{figure*}[t]
    \centering
    \includegraphics[width=0.85\linewidth]{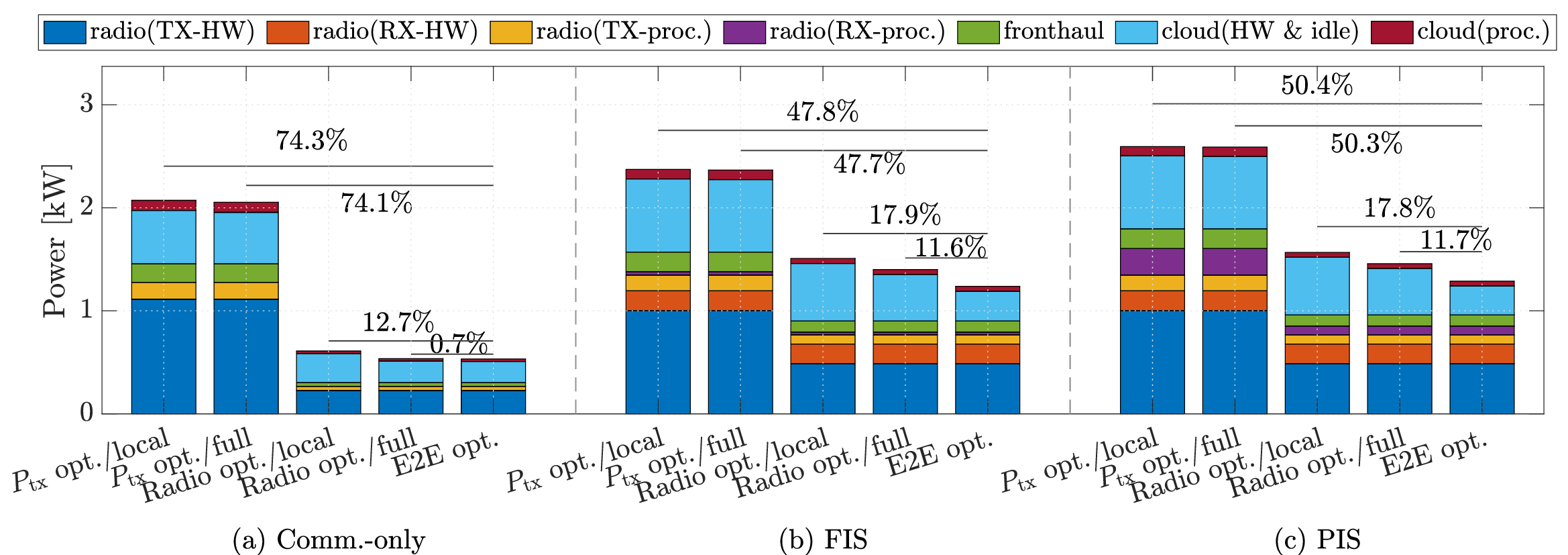}\vspace{-2mm}
    \caption{Power consumption for $L=25$ and $K=8$ in (a) communication-only, (b) FIS, and (c) PIS settings.}
    \label{fig:breakdown}\vspace{-5mm}
\end{figure*}

\vspace{-3mm}
\section{Numerical Results}
\label{sec:results}
We consider a $500\,\mathrm{m}\times500\,\mathrm{m}$ area containing four SSAs centered at $(125,125)$, $(125,375)$, $(375,125)$, and $(375,375)$\,m. Unless otherwise stated, $25$ APs, each equipped with four antennas, are deployed on a grid, while eight single-antenna UEs are randomly distributed over the area. The maximum downlink transmit power per AP is $1$\,W, and the uplink pilot power is $0.2$\,W. All targets have an RCS variance of $\sigma_{\rm rcs}=-5$\,dBsm.
The false-alarm probability is set to $0.03$, following the requirements for 3GPP sensing service category~4 object-detection and tracking scenarios \cite[Table~6.2-1]{3gpp_ts_22_137}, and the sensing blocklength is $\tau_s=20$. The large-scale fading, LOS probability, and Rician factors follow the 3GPP Urban Microcell model \cite[Tables~B.1.2.1-1, B.1.2.1-2, and B.1.2.2.1-4]{3gpp2010further}. The power-model parameters are adopted from \cite[Table~1]{demir2023cell}, while the exponent $v$ in the distributed sensing weights \eqref{eq:weights} is set to $0.25$ following \cite{behdad2025detecting}.
For Algorithm~\ref{algo:E2E}, the convergence tolerance is $\epsilon=0.1$, and the maximum number of iterations is $10$. The slack-penalty weight is set to $\lambda_0=10^3$. The binary-penalty weights are initialized as $5$ and $100$ for $\lambda_1-\lambda_4$ and $\lambda_5$, respectively. After each iteration, they are multiplied by $3$ until they reach a maximum value of $500$.
All recovered solutions are checked against $\mathcal{P}_0$; for $L=25$ and $K=8$, the feasibility rate is $86.7\%$ over 30 realizations.

Fig.~\ref{fig:breakdown} compares the average network power consumption for communication-only, FIS, and PIS settings, where communication-only represents conventional CF-mMIMO without sensing. Integrating sensing increases the total power consumption due to the additional transmission, processing, and fronthaul requirements, with an increase of up to approximately 535 W under Ptx opt./local. The proposed E2E optimization substantially reduces this overhead, lowering the total power to 1.24 kW and 1.29 kW for FIS and PIS, respectively. This corresponds to power savings of 47.8\% and 50.4\% over Ptx opt./local, and 11.6\% and 11.7\% over radio opt./full, demonstrating the benefit of jointly orchestrating radio, fronthaul, and cloud resources.


Figs.~\ref{fig:tot_power_fis} and \ref{fig:tot_power_pis} show the impact of AP density and traffic load on total power consumption for $L=25,37$ APs and $K=4,8$ UEs. Benchmarks with feasibility ratios below $70\%$ are excluded, and the averages are computed over setups jointly feasible for the remaining schemes. The proposed E2E framework achieves the lowest power consumption across all configurations, demonstrating its robustness to different deployments and traffic loads. Increasing $K$ increases the power consumption, whereas increasing $L$ can reduce it by providing more favorable AP choices and requiring fewer active APs, particularly under the full-coordination and E2E schemes. For PIS with $L=37$, the transmit-power benchmarks are excluded because their energy-unaware associations frequently violate the processing constraints.

\begin{figure}[!t]
    \centering
    \begin{subfigure}[b]{0.65\linewidth}
        \centering
        \includegraphics[width=\linewidth]{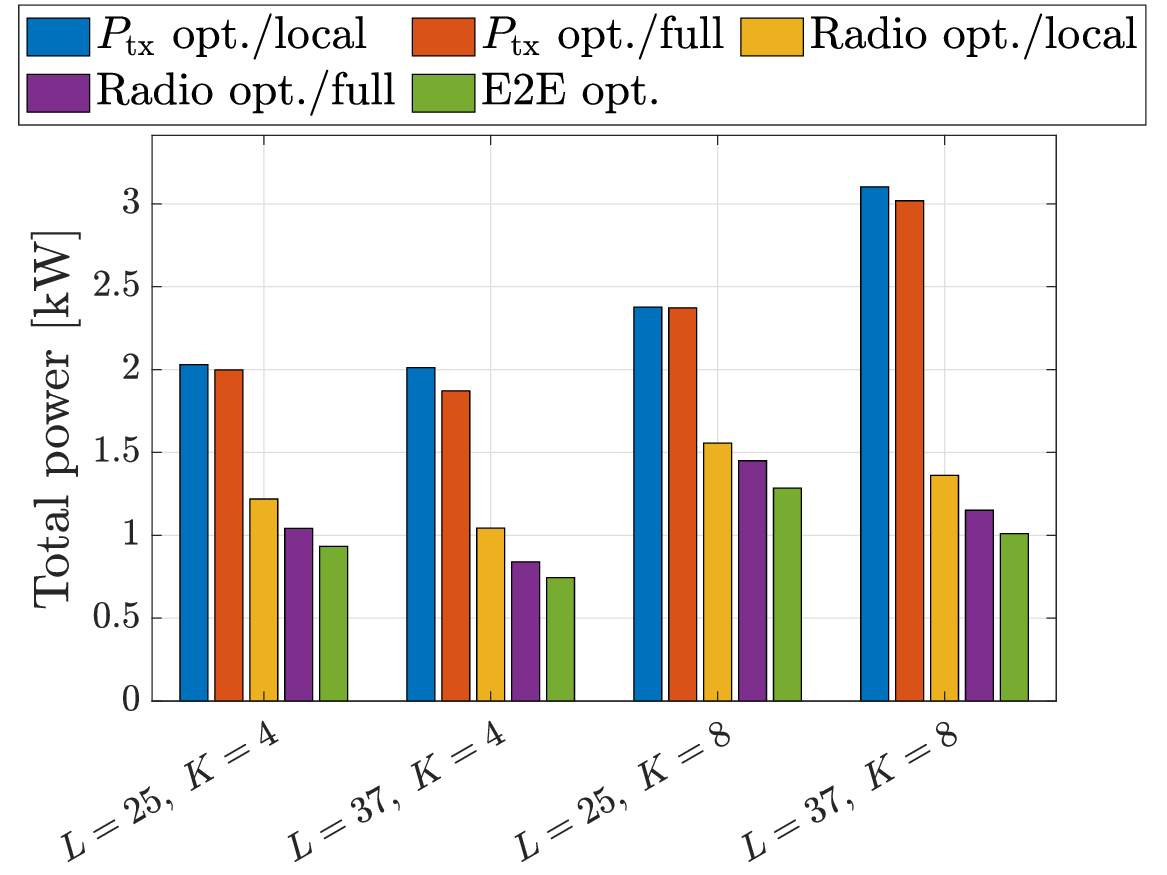}\vspace{-2mm}
        \caption{FIS}
        \label{fig:tot_power_fis}
    \end{subfigure}
    \hfill
    \begin{subfigure}[b]{0.65\linewidth}
        \centering
        \includegraphics[trim={0mm 0mm 0mm 28mm},clip,width=\linewidth]{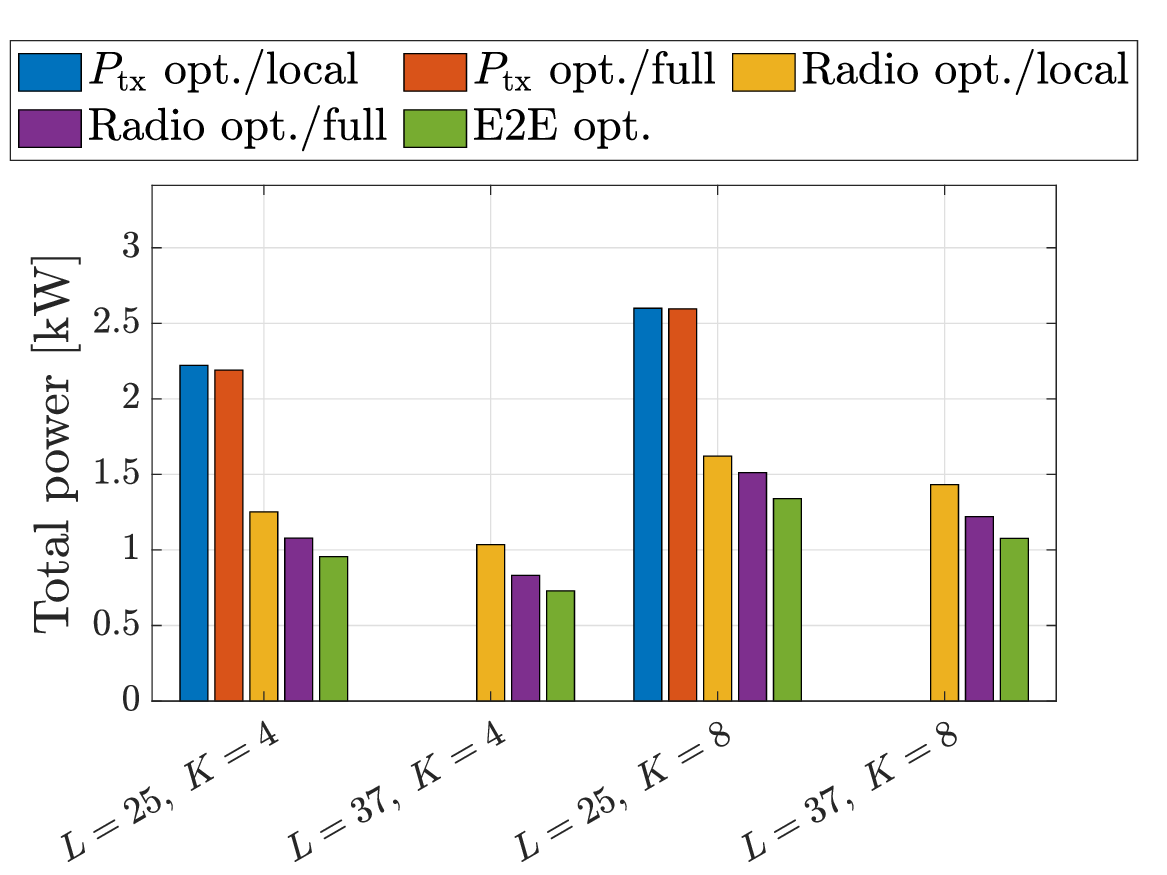}\vspace{-2mm}
        \caption{PIS}
        \label{fig:tot_power_pis}
    \end{subfigure}
    \caption{Total power consumption under (a) FIS and (b) PIS.}
    \label{fig:tot_power}\vspace{-6mm}
\end{figure}

\begin{figure}[t]
    \centering
   \includegraphics[width=0.65\linewidth]{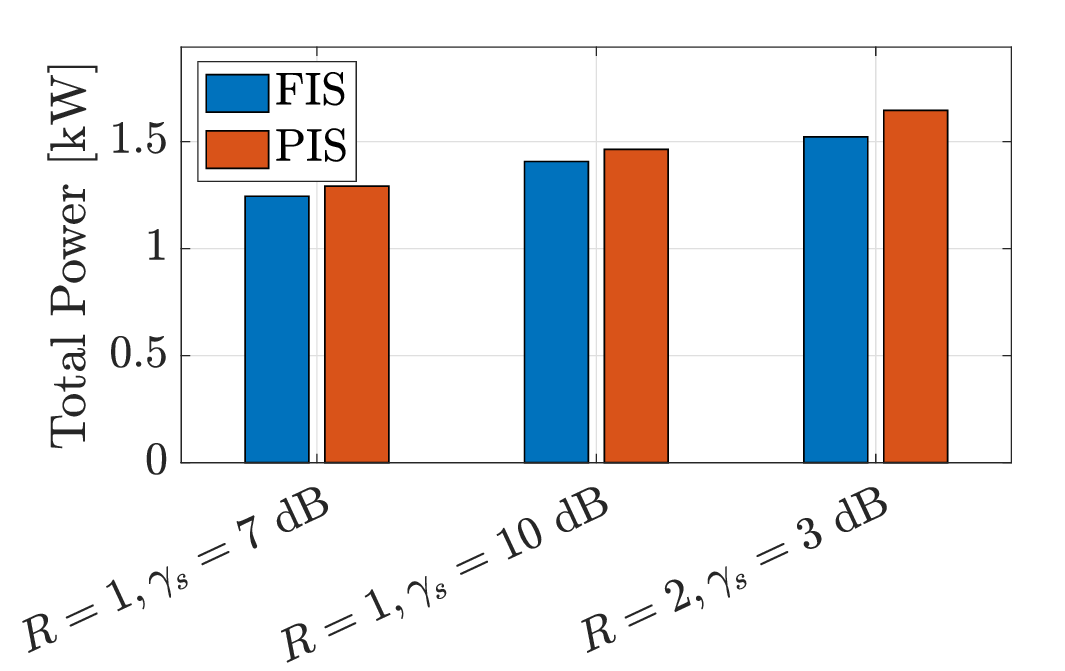}
    \caption{Total power consumption under E2E optimization.}
    \label{fig:R}\vspace{-3mm}
\end{figure}

Fig.~\ref{fig:R} shows the total power consumption for different
$(R,\gamma_s)$ configurations. For $R=1$, increasing $\gamma_s$ from
$7$ to $10$\,dB increases the power consumption in both FIS and PIS. A further
increase is observed for $(R,\gamma_s)=(2,3\,\mathrm{dB})$, mainly due to the
additional active RX-APs and the transmit power required to satisfy the sensing
constraints. Table~\ref{tab:detection} shows that FIS consistently achieves
higher detection probability than PIS, although the performance gap becomes
smaller when the sensing requirement is stricter or more RX-APs are assigned
to each SSA.

Fig.~\ref{fig:all} shows the detector-processing load and fronthaul data-rate requirement versus the number of SSAs. The processing load increases with $S$, $L_{\rm tx}$, and $\tau_s$. Due to its iterative estimation procedure, PIS requires substantially more processing than FIS, approaching $10^4$ GOPS for $S=10$, $L_{\rm tx}=20$, and $\tau_s=200$. Conversely, FIS requires a higher fronthaul rate because the instantaneous transmit-signal vectors are shared for all sensing symbols, causing its fronthaul load to increase with both $\tau_s$ and $L_{\rm tx}$. PIS shares only covariance matrices, making its fronthaul requirement independent of $\tau_s$, although it scales quadratically with $L_{\rm tx}$. Thus, FIS provides lower processing complexity and better detection performance, whereas PIS substantially reduces the fronthaul requirement.

\begin{table}[!t]
\caption{Detection probability under E2E optimization.}
\centering
\setlength{\tabcolsep}{5pt}
\renewcommand{\arraystretch}{1.1}\vspace{-1mm}
\begin{tabular}{|c|c|c|c|}
\hline
 & $R=1,\gamma_s=7$\,dB & $R=1,\gamma_s=10$\,dB & $R=2,\gamma_s=3$\,dB \\ \hline
FIS & 0.984 & 0.995 &  0.999  \\ \hline
PIS & 0.956 & 0.987 & 0.998  \\ \hline
\end{tabular}\vspace{-3mm}
\label{tab:detection}
\end{table}

\begin{figure}[t]
    \centering
    \includegraphics[width=0.85\linewidth]{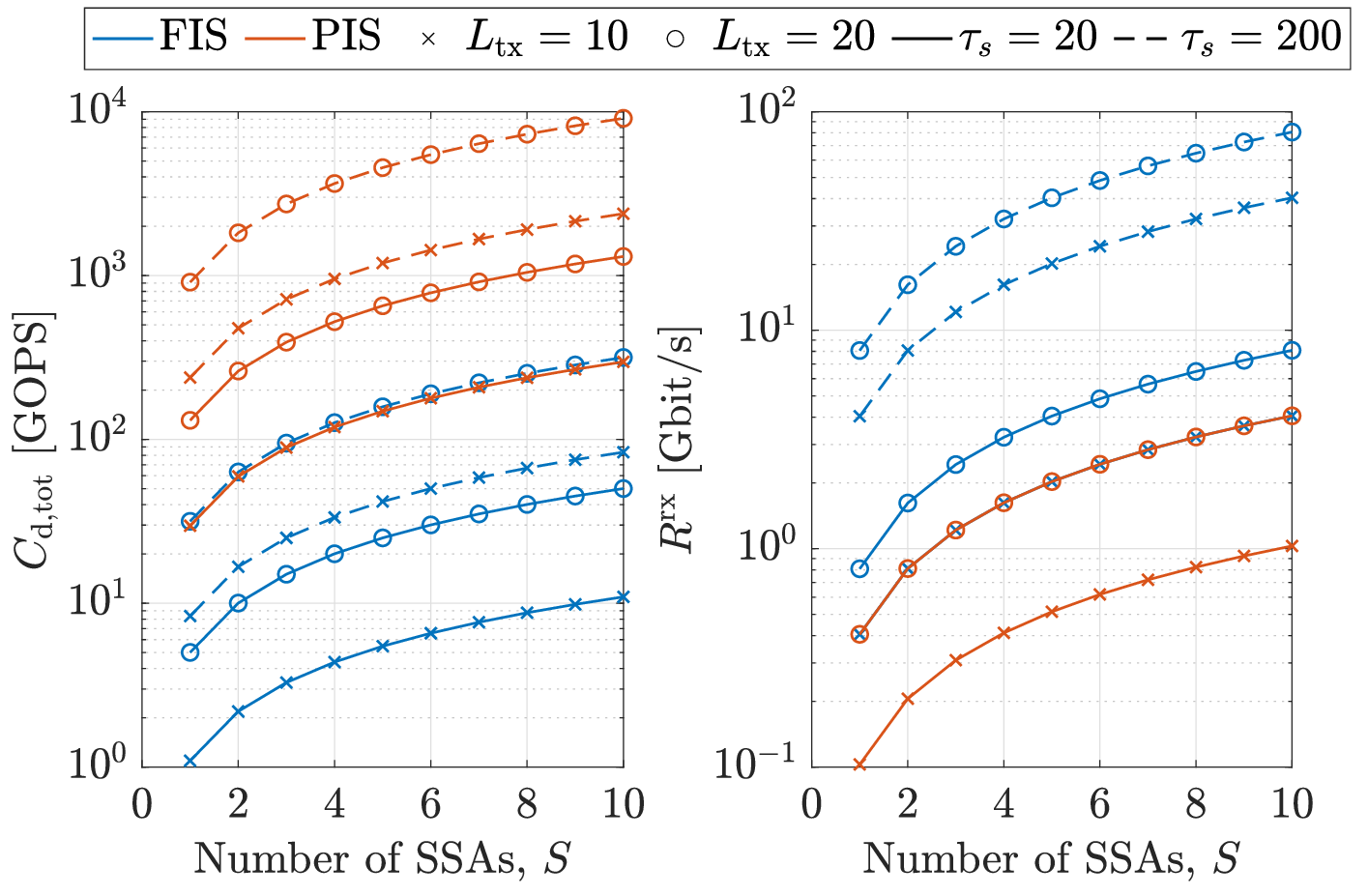}
    \caption{Required GOPS for detector, and fronthaul data rate per RX-AP, vs number of SSAs under FIS and PIS.
    }
\label{fig:all}\vspace{-6mm}
\end{figure}

We further compare the proposed detector with a no-interference baseline obtained by setting $I_{s,r}[m]=0$. For $L=25$, $K=8$, and $R=1$, the average detection probability decreases only slightly, from $0.989$ to $0.984$, when interference is included, confirming that the approximation has a limited impact in the considered regime.

\section{Conclusion}
\label{sec:conclusion}
We developed a cross-layer E2E resource-orchestration framework for energy-efficient distributed multi-target sensing in CF-mMIMO ISAC systems. RX-APs compute local MAPRT statistics that are aggregated at the cloud, with FIS and PIS capturing different levels of transmit-signal information. Their processing and fronthaul requirements were incorporated into an E2E power model, and a mixed-integer non-convex problem was formulated to jointly optimize radio, fronthaul, and cloud resources under communication, sensing, and resource-capacity constraints. A two-stage iterative solution based on Big-$M$ reformulation, FPP-SCA, CCP, penalty-based relaxation, structured binary recovery, and AP refinement was developed.
Numerical results demonstrated that integrating sensing can considerably increase network power consumption when resources are not jointly orchestrated. The proposed E2E framework offered power savings up to $50\%$, compared with transmit-power-only optimization under local coordination. It also achieved approximately $11-17\%$ savings over radio optimization benchmarks. These reductions were obtained while maintaining detection probabilities above $0.95$ in the considered scenarios.
The results also reveal key system-level trade-offs: FIS provides lower detector-processing complexity and slightly better detection performance at the cost of higher fronthaul requirements, whereas PIS reduces fronthaul signaling but requires greater local processing. Stronger sensing requirements and additional RX-APs improve detection at the expense of higher network power, while denser AP deployments can reduce power through more favorable associations. Overall, the results highlight the importance of jointly orchestrating radio, fronthaul, and cloud resources for scalable and energy-efficient CF-mMIMO ISAC.

\appendices
\section{Proof of Proposition~\ref{prop:gops}}
\label{app:gops}
For the FIS detector, each RX-AP computes the vector $\mathbf a_{s,r}$ in \eqref{eq:vecA}, the inverse matrix $\mathbf C_{s,r}^{-1}$ in \eqref{eq:matrixCinv}, and the test statistic $T_{s,r}^{\rm FIS}$ in \eqref{eq:T_r_FIS}. Let
$L_{\rm tx}=\sum_{j=1}^{L}z_j$. Computing $\mathbf a_{s,r}$ requires $8\tau_sL_{\rm tx}$ operations. Constructing
$\sum_{m=1}^{\tau_s}\mathbf c_{s,r}[m]\mathbf c_{s,r}^{H}[m]$
requires $4\tau_sL_{\rm tx}^{2}$ operations by exploiting its Hermitian structure. The cost of computing
$\sigma_n^2(\mathbf R_{s,r}^{\rm rcs})^{-1}$ is neglected because this term remains constant over a long period. The matrix inversion requires
$\frac{8}{3}L_{\rm tx}^{3}-\frac{8}{3}L_{\rm tx}$ operations. Finally, computing $T_{s,r}^{\rm FIS}$ requires
$8L_{\rm tx}^{2}-2L_{\rm tx}$ operations. Hence,
\begin{align}\label{eq:C_FIS}
    C_{\textrm{maprt}}^\textrm{FIS}&= \phi \Bigg(\frac{8}{3} L_{\rm tx}^3 + \Big(4\tau_s + 8\Big)L_{\rm tx}^2 + \Big(8\tau_s -\frac{14}{3}\Big)L_{\rm tx}\Bigg)
\end{align}
where $\phi=\frac{N_\textrm{used}}{T_\textrm{s}\tau_c 10^9}$. 
The operations associated with multiplication by $\sqrt M$ and $M$ in \eqref{eq:vecA} and \eqref{eq:matrixCinv} are omitted because these scaling factors cancel when evaluating \eqref{eq:T_r_FIS}.
Computing the three terms of $T_{s,r}^{\rm PIS,(i)}$ in
\eqref{eq:fpartlyinformed} requires
$8L_{\rm tx}^{2}-2L_{\rm tx}$,
$8L_{\rm tx}$, and
$8\tau_sL_{\rm tx}^{2}-2\tau_sL_{\rm tx}$ operations, respectively. The intermediate quantities are reused from the computation of
$\hat{\boldsymbol\alpha}_{s,r}$. Therefore, the processing load per PIS iteration is
\begin{align}
C_{\rm maprt}^{\rm PIS,(i)}
=
\phi\Bigg(
\frac{16}{3}L_{\rm tx}^{3}
+\Big(12\tau_s+24\Big)L_{\rm tx}^{2}
+\Big(18\tau_s+\frac{14}{3}\Big)L_{\rm tx}
\Bigg).
\label{eq:C_PIS}
\end{align}
Since the PIS detector requires $n_{\rm iter}$ iterations, its total load is
$n_{\rm iter}C_{\rm maprt}^{\rm PIS,(i)}$. Substituting
\eqref{eq:C_FIS} and \eqref{eq:C_PIS} into
\eqref{eq:gops_rx_general} and collecting the cubic, quadratic, linear, and constant terms yields
\eqref{eq:gops_rx_compact}, completing the proof of Proposition~\ref{prop:gops}.



\section{Proof of Proposition~\ref{prop:cloud_complexity}}
\label{app:cloud_complexity}
Let
$L_{\rm tx}=\sum_{l=1}^{L}z_l$,
$U_k=\sum_{l=1}^{L}\eta_{k,l}$, and
$V_t=\sum_{l=1}^{L}\zeta_{t,l}$.
The communication-related cloud processing load is 
$C_{\rm cloud}^{\rm comm}
=F_{\rm fixed}
+C_{\rm other\text{-}tx}L_{\rm tx}
+C_{\rm other\text{-}ue}
\sum_{k=1}^{K}U_k,$ 
where $C_{\rm other-tx}$ and $ C_{\rm other-ue}$ denote the constant GOPS per TX-AP and UE association, respectively. $F_{\rm fixed}$ is a fixed GOPS load that depends on the number of UEs $K$, required SE, and $N_{\rm bits}$.

Aggregating the local test statistics requires $\phi\sum_{l=1}^{L}\sum_{s=1}^{S}\xi_{s,l}$ GOPS, where $\phi=N_{\rm used}/(T_s\tau_c10^9)$.

For FIS, computing the transmitted signal components over $\tau_s$ sensing symbols requires
$(8\tau_sM+4)
\left(
\sum_{k=1}^{K}U_k
+\sum_{t=1}^{S}V_t
\right)$
operations. Computing the array projections for all $S$ SSAs
and applying the large-scale channel coefficients for the $RS$
target--RX-AP associations further requires $(8\tau_sMS+4\tau_sRS)L_{\rm tx}$ operations. Hence, the FIS cloud processing load is
\begin{align}
C_{\rm c}^{\rm FIS}
={}&\phi
\Big(\!
(8\tau_sM+4)
(
\sum_{k=1}^{K}U_k
+\sum_{t=1}^{S}V_t
)\nonumber\\[-1mm]
&
+(8\tau_sMS+4\tau_sRS)
L_{\rm tx}\!
\Big).
\label{eq:proof_cloud_fis}
\end{align}

For PIS, $S(8M+4)
\Big(\sum_{k=1}^{K}U_k
+\sum_{t=1}^{S}V_t
\Big)$ operations are required for computing the elements of $\mathbf{u}_{s,k}$ and
$\mathbf{v}_{s,t}$ for all $S$ SSAs. Constructing
$\widetilde{\mathbf R}_s$ requires
$4S
\sum_{k=1}^{K}U_k^2
+4S\sum_{t=1}^{S}
V_t^2$
operations by exploiting the Hermitian structure of the matrix. The large-scale channel coefficients are assumed
to be precomputed, and their complexity is therefore neglected.
Furthermore, computing each $\mathbf R_{s,r}$ requires
$2L_{\rm tx}^{2}$ operations, resulting in
$2RSL_{\rm tx}^2$
operations for all target--RX-AP associations. Consequently, the corresponding load is
\begin{align}
C_{\rm c}^{\rm PIS}
&=\phi
\Big(
S(8M+4)
(\sum_{k=1}^{K}U_k
+\sum_{t=1}^{S}V_t
)+4S\sum_{k=1}^{K}
U_k^2\nonumber\\
&+4S\sum_{t=1}^{S}
V_t^2+2RSL_{\rm tx}^2
\Big).
\label{eq:proof_cloud_pis}
\end{align}
Finally, the cloud processing load after collecting the corresponding terms is obtained as
$C_{\rm cloud}
\!=\!C_{\rm cloud}^{\rm comm}
+\phi\sum_{l=1}^{L}\sum_{s=1}^{S}\xi_{s,l}
+\mathbb I_{\rm FIS}C_{\rm c}^{\rm FIS}
+(1-\mathbb I_{\rm FIS})C_{\rm c}^{\rm PIS}
$,
which yields \eqref{eq:C_cloud_total} thereby completing the proof.

\bibliographystyle{IEEEtran}
\bibliography{refs}
\end{document}